%% file: article.tex
\documentclass[oribibl]{llncs}
\usepackage{amsmath,amsfonts,amssymb}
\usepackage{bm}
\usepackage{bbm}%for \bbOne
\usepackage{thm-restate}
\usepackage[inference]{semantic}
\usepackage{tikz}
\usetikzlibrary{arrows,automata,shapes.geometric,positioning,calc}
\tikzset{sd/.style={state,diamond,minimum size=5mm,inner sep=0pt}}
\tikzset{so/.style={state,circle,inner sep=0pt,minimum size=5mm}}
\usepackage{wrapfig}
\usepackage{chngcntr}
\usepackage{microtype}
\usepackage{xifthen}
\usepackage{hyperref}
\usepackage{learning}
\usepackage[noend]{algorithm2e}
\usepackage{graphicx}
\usepackage{subcaption}
\newcommand{\myorcid}[1]
{\unskip\texorpdfstring{
\href{https://orcid.org/#1}{\protect\includegraphics[width=10px]{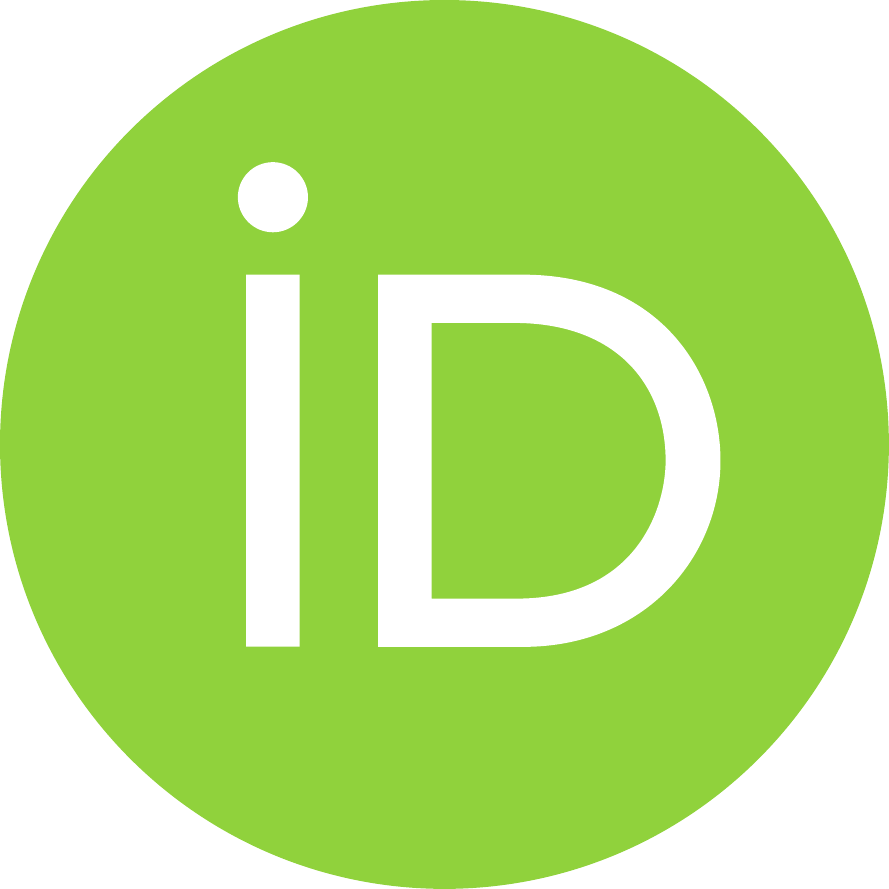}}}{}
}
\newif\iflong

\longtrue %defines \iflong as true

\title{Active learning of timed automata with unobservable resets%
  \thanks{This work was partially funded by ANR project
    Ticktac (ANR-18-CE40-0015).}}
\author{L\'eo Henry \myorcid{0000-0001-6778-5840} \and
  Thierry J\'eron \myorcid{0000-0002-9922-6186} \and
  Nicolas Markey \myorcid{0000-0003-1977-7525}}
\institute{Univ Rennes, Inria, CNRS -- Rennes, France \\
\email{firstname.lastname@inria.fr}}

\begin{document}
    \maketitle

\begin{abstract}
  \input{abstract}
\end{abstract}

\section{Introduction}
\input{intro}
\section{Preliminaries}
\label{sec:prelim}
\input{prelim}
\section{Observation structure}
\label{sec:structure}
\input{structure}
\section{Updating a timed observation structure}
\label{sec:updates}
\input{updates}
\section{Building a candidate timed automaton}
\label{sec:rebuild}
\input{rebuild}
\section{Conclusion}
\input{conclu}
\bibliographystyle{alpha}
\bibliography{biblio}
% \printbibliography

\iflong
\clearpage
\appendix
\section{Algorithms}
\label{app:algos}
\input{appendixalgos}
\section{Proofs}
\label{app:proofs}
\input{appendixproofs}
\else
\fi
\end{document}

%% file: abstract.tex
Active learning of timed languages is concerned with the inference of timed
automata
%from observed timed words.
by observing some of the timed words in their languages.
The~learner can query for the
membership of words in the language, or propose a candidate
model and ask if it is equivalent to the target. The~major difficulty
of this framework is the inference of \emph{clock resets}, which are
central to
the dynamics of timed automata but not directly observable.

Interesting first steps have already been made by
restricting to the subclass of \emph{event-recording automata}, where
clock resets are tied to observations. In~order to advance towards
learning of general timed automata, we~generalize this method to a new
class, called \emph{reset-free} event-recording automata, where some
transitions may reset no clocks. 
% This~offers the same challenges as
% generic timed automata while keeping the simpler framework of
% event-recording automata for the sake of readability.
%for our explanations.

Central to our contribution is the notion of \emph{invalidity}, and
the algorithm and data structures to deal with~it, allowing on-the-fly
detection and pruning of reset hypotheses that contradict
observations. 
This notion is a key to any efficient active-learning
procedure for generic timed automata.

%% previous version
%%------------------
% The active learning focuses on learning through interactions with a system -called a teacher. The agent can query for the membership of words in the target language, or propose a candidate model to be verified to be equivalent or not to the target. This framework has a great learning power, and has been well studied for untimed languages.

% One interesting extension is to consider \emph{continuous time} in the language to be learned. A classic model for such languages is the timed automaton (TA). The main challenges that arise in this framework are the continuous behavior due to time, that affects sampling, and the unobservability of \emph{resets}, that are central to the dynamic of a TA.

% In this work, we generalize an existing procedure to handle the resets using the \emph{invalidity} notion , and propose a data structure and algorithms that deal with that more general class of TA.

%% file: intro.tex
Active learning~\cite{Ang87a} is a type of learning
%that
%was first proposed by Angluin~\cite{Ang87}.  It
%uses
in which
a teacher assesses the learner's
progress and direct the learning effort toward meaningful
decisions. The~learner can request information from the teacher
%in the form of
via \emph{membership queries}, asking about a specific observation,
and \emph{equivalence queries}, proposing to compare the current hypothesis
to the correct model; in the latter case, the~teacher either accepts
the hypothesis or returns
a counter-example exemplifying mispredictions of the learner's
hypothesis.

This framework is well-studied in the setting of finite-state
automata~\cite{Ang87a,Ang87b,Ang90}, and allows to make sound proofs
for both correctness and complexity of learning algorithms. As~most
real-life systems dispose of \emph{continuous} components,
attempts have been made to leverage this framework to take
them into account.  One~of the most classic additions
is \emph{time}. An~observation is then a timed word, made of actions and
delays between them. One of the most recognized models for such timed
languages is the timed automaton~(TA), but its dynamics are complex:
TAs~measure time using a set of \emph{clocks} that hold a positive
real value progressing with time, can be compared with integer constants 
to allow or disallow
transitions, and reset to zero along those transitions. For a learning
algorithm, one of the main challenges is to deal with those resets,
that are typically not observable, but play a central role in the
system dynamics.

Some work has already been done in the active learning of subclasses of~TAs,
mostly deterministic TAs with only one clock~\cite{ACZZZ20} and
deterministic event-recording automata (DERA)~\cite{GJP06,Gri08}, which
have as many clocks as actions in the alphabet, and where each clock
encodes exactly the time elapsed since the last corresponding action
was taken. These classes of automata present the advantages of having
a low-dimensional continuous behaviour (for~1-clock~TAs) and to allow
to derive the resets of the clocks directly from the observations (for~DERA).
Other approaches have been investigated for the learning of timed
systems. Learning of TAs from tests has been studied using genetic
algorithms~\cite{TALL19}, which is a very different approach to ours
to exploit a similar setting. Inference of simple TAs from positive
data~\cite{VMC08,VMC12} has also been well studied. These works are
more loosely related to ours, as our setting greatly differs from
positive inference.

We propose in this work to generalize to a class of timed automata
enjoying both several clocks and different possible resets
that can not be inferred directly from observations. 
This allows us to design and prove algorithms that handle all the main 
difficulties that arise in deterministic TAs, making this contribution 
an important first step towards active learning for generic deterministic TAS.
% We argue that our
% method can be lifted up to all \emph{deterministic} TAs with a known
% number of clocks, albeit to a greater computational cost, making this
% contribution an important first step towards active learning for
% generic deterministic TAs.

%
To our knowledge, the closest works are Grinchtein's thesis on
active learning of~DERA~\cite{Gri08} and the paper proposing to learn
one clock TAs~\cite{ACZZZ20}. The~work of
Grinchtein~\emph{et~al.}~\cite{GJP06} is
the most related to ours, as we use some of the data structures they
developed and keep the general approach based on timed decision
trees.
The main difference between our work and this one is that we handle
the inference of resets in a class of models in which they can not
be \emph{directly} deduced from observations. The~approach reported
in~\cite{ACZZZ20} proposes to deal with reset guessing, but makes it
in a somewhat "brute force" manner, by directly applying a
branch-and-bound algorithm and jumping from model to~model.  In~order
to be able to deal with larger dimensions, \eg~to~handle TAs with a
large set of clocks, we need to be more efficient by exploiting the
theory built around TAs and detecting invalid models as early as possible.

\iflong
For reasons of space, the proofs of our claims and the pseudo-code of our algorithms are left in appendix.
\else 
For the details of the proofs and algorithms, we refer the reader to the long version~\cite{arxiv}.
\fi

%% file: prelim.tex
%In this section, we first define timed automata and the subclass of
%event-recording automata. We~then introduce language-related concepts
%that will be useful in our learning setting.

%In the following section, we formally define timed automata and the
%subclasses we~consider in this paper. We~then  spend some time defining
%different abstractions of timed words and their correspondence with
%timed automata.

%\subsection{Clocks and valuations}
\subsection{Timed automata}

%First of all, we define \emph{clocks}, that are used to introduce time
%in our models, and some useful abstractions and operations on them.
For the rest of this paper, we~fix a finite
alphabet~$\Sigma$.

Let \(\Clocks[]\) be a finite set of variables called
\emph{clocks}. A~valuation for~$\Clocks[]$
is a function \(v\colon\Clocks[]\rightarrow
\bbR_{\geq0}\).
%We~will often abuse this notation and consider a
%valuation as a an element of \(\bbR_{\geq0}^{|\Clocks[]|}\).
We write \(\initv\) for
the clock valuation associating \(0\) with all clocks.
%
%We define three operations on clock valuations: delay, future, and reset.
For any \(\delta\in\bbR_{\geq0}\) and  any valuation \(v\)
we write $v+d$ for the valuation such that \((v+\delta)(x)= v(x)+\delta\)
for each clock~$x$; this  corresponds to elapsing~\(\delta\)
time units from valuation~\(v\).
%
%We also define the \emph{future operation},
%\(v\future\bigcup_{t\in\bbR_{\geq0}}{v+t}\)
%which associates to
The~future of a valuation~$v$ is the set 
$v\future = \{v+t \mid t\in\bbR_{\geq 0}\}$ of its time successors.
Finally, for any \(\Clocks[]'\subseteq \Clocks[]\) and any valuation \(v\),
we~write $v_{[\Clocks[]'\leftarrow0]}$ for the valuation such that
$v_{[\Clocks[]'\leftarrow0]}(x)=v(x)$ for all~$x\notin \Clocks[]'$
and $v_{[\Clocks[]'\leftarrow0]}(x)=0$ for all~$x\in \Clocks[]'$.

\looseness=-1
Simple clock constraints are expressions of the forms 
%We define constraints on the clocks as relations
\(x-x'\thicksim n\) and \(x\thicksim n\), for \(x,x'\in \Clocks[]\),
\(\mathord\thicksim\in\{\mathord<,\mathord\leq,\mathord=,\mathord\geq,
\mathord>\}\) and \(n\in\bbN\).  We call \emph{zone} over~$\Clocks[]$
any finite conjunction of such constraints, and write 
\(\Zones[{\Clocks[]}]\) for
the set of zones over~\(\Clocks[]\).  Given a valuation~$v$ and a
zone~$z$, we~write $v\models z$ when $v$ satisfies all the constraints
in~$z$.  We~may identify a zone~\(z\) with the set of valuations~$z$
such that $v\models z$.  We~call \emph{guard} any zone not involving
constraints of the form \(x-x'\thicksim n\), and write 
\(\Guards[{\Clocks[]}]\) for the set of guards.
We~extend all three operations on valuations to zones elementwise.
%% subset of \(\bbR_{\geq0}^{\Clocks[]}\)
%% that satisfies its constraints.  We call the subset of zones that do
%% not use diagonal constraints \emph{guards} and note their set
%% \(\Guards[]\). Notice that the future operator defines a zone from a
%% valuation, and that all three operators defined above
%% can be naturally extended to zones.
%% %
%For a valuation \(v\) and a zone \(z\), we say that \(v\) satisfies \(z\), noted \(v\models z\) when \(v\in z\).

%\subsection{Timed automata}
% In this subsection, we define timed automata (TA) and some subclasses
% of interest for our contribution. 
%TA are closely related to the classes of words defined before. A \emph{run} in a timed automaton corresponds to a timed word with resets, while a \emph{path} is a guarded word with resets. 
\begin{definition}%{Timed automaton}{ta}
  A \emph{timed automaton} (TA) over~\(\Sigma\) is a tuple
  \(\ta=\tuple{\Loc,\initloc,\Clocks[],\Trans,\Accept}\) such that:
  %\begin{itemize}
    %\item 
    \(\Loc\) is a finite set of \emph{locations}, and \(\initloc\in\Loc\) 
    is the \emph{initial location};
    %\item 
    \(\Clocks[]\) is a finite set of clocks;
    %\item 
    \(\Accept\subseteq\Loc\) is a set of accepting locations;
    %\item 
    \(\Trans\subseteq\Loc\times\Sigma\times\Guards[{\Clocks[]}]\times 2^{\Clocks[]}\times\Loc\) is a set of transitions. 
    For a transition \((\loc,a,g,r,\loc')\), we call \(g\) its guard, \(a\) its action and \(r\) its reset.
  %\end{itemize}
\end{definition} 
%From this definition it is already clear that a \emph{path} in a timed automaton, \ie~a sequence of consecutive transitions, corresponds to a guarded word with resets. 
%We call a
We~write $K_\ta$ (or~$K$ when the context is clear) for
the~\emph{maximal constant} appearing in~$\ta$.
%noted \(K_\ta\) or \(K\) when the context is clear, the greatest
%integer appearing in guards on the transitions of this TA.
%
%A~\emph{path} of a timed automaton is a sequence of consecutive transitions.
%\NM{ai d\'eplac\'e la def de ``path in a TA''. }
%
%
We say that a TA is \emph{deterministic} when, for any two transitions
\((\loc,a,g,r,\loc')\) and \((\loc,a,g',r',\loc'')\)
%, we have that
where \(g\wedge g'\) is satisfiable, it~holds $\loc'=\loc''$ and $r=r'$. 
% In the following, we will only consider deterministic~TAs. 
We~only consider deterministic TAs in the sequel,
as active-learning methods can only target this (strict) subclass of~TAs.

%% \begin{remark}%{Determinism}{det}
%%   Requiring determinism \emph{is} a restriction, as non-deterministic
%%   TAs are strictly more expressive than~TAs~\cite{AD94}. It~is enforced because
%%   of the way active-learning methods construct a model from
%%   observations of its languages, which is inherently deterministic.
%% \end{remark}

%A~\emph{path} of a timed automaton is a sequence of consecutive
%transitions. However, besides plain transitions, TA also have
%\emph{delay transitions}, which are implicit in the model.
%The semantics of TAs is usually defined as the infinite
%(untimed) transition system of its configurations.
\begin{definition}%{Semantics of a TA}{ta_sem}
  With a TA
  \(\ta=(\Loc,\initloc,\Clocks[],\Trans,\Accept)\), we~associate
  the transition
  system
  \(\semta=(\Conf=\Loc\times\bbR_{\geq0}^{|\Clocks[]|},(\initloc,\initv),\semT,\Accept_{\semta})\)
  where \(\Loc\times\bbR_{\geq0}^{|\Clocks[]|}\) is the set of
  \emph{configurations}, \((\initloc,\initv)\) is the initial
  configuration, \(\Accept_{\semta}=\{(\loc,v)\mid \loc\in\Accept\}\)
  is the set of accepting configurations, and \(\semT\subset
  \Conf\times(\bbR_{\geq0}\cup\Trans)\times\Conf\) a set of
  transitions, such that for any $(\loc,v)\in \Conf$:
  %\begin{itemize}
  %\item
  (a)~for any~$\delta\in\bbR_{\geq 0}$, we~have
    $((\loc,v),\delta,(\loc,v+\delta))$ in~$\semT$;
  %\item
  (b)~for any $e=(\loc,a,g,r,\loc')\in \Trans$ s.t. $v\models g$, we have
    $((\loc,v),e,(\loc',v_{[r\leftarrow 0]}))$ in~$\semT$.
  %\end{itemize}
  %\LH{Je pense qu'introduire les variables
  %  prendrait plus de place...}
%  \[
%  \inference[$\delta$]{(\loc,v)\in\Conf,\, \delta\in\bbR_{\geq0}}{
%    ((\loc,v),\delta,(\loc,v+\delta))\in\semT
%  }\ 
  % \]
  % \[
%  \inference[$\Trans$]{(\loc,v)\in\Conf,\, e=(\loc,a,g,r,\loc')\in\Trans,\, 
%  v\in g}{
%    ((\loc,v),e,(\loc,v[r\leftarrow0]))\in\semT
%  }  
%  \]
\end{definition}

A~path in a timed automaton~$\ta$ is a sequence of transitions in the associated
transition system~$\semta$.
%
%The semantics of a TA gives its link to timed words. 
A~\emph{timed word with resets} of~\(\ta\) is a path
\(w_{tr}=((l_i,v_i)\xrightarrow{e_i}(l_{i+1},v_{i+1}))_{i\in[0,n]}\in(\Conf\times(\semT\cup\bbR_{\geq0}))^{*}\times\Conf\)
of its semantics \(\semta\).
%\footnote{As we consider only deterministic TA, we
%will often give the letter and reset corresponding to the transition
%instead of the transition itself.}.
A~timed word with resets is \emph{accepting} when its
final configuration is in \(\Accept_{\semta}\).

%% In the process of learning, we are given a set of \emph{observations}
%% of such runs, together with the information whether each of these
%% observations is accepted (which we denote with~$\ltrue$) or not
%% (denoted with~$\lfalse$). Observations include actions and delays, but
%% they do not include informations about the values of the clocks; in
%% particular, clock resets are not observed.
%% \NM{Ce paragraphe n'est pas au bon endroit}
%% \LH{Il fait un peu bizarre, mais où le mettre?}
%% En fait c'est dit \`a plusieurs endroits dans la suite...
%
%In this work, we consider that only actions and delays are
%observable from TA runs. Hence the observation of a run is its
%projection from transitions to letters of the alphabet. We furthermore
%observe a boolean telling us if a run is accepting, that will be noted
%\(\ltrue\) for "yes" and \(\lfalse\) for "no".

In order to obtain a finite representation of the infinite set of
timed words with resets, we use an abstraction based on the following notion of
\emph{$K$-equivalence}.
\begin{definition}%{K-equivalence}{}
Two nonnegative reals $x$ and $y$ are \emph{$K$-equivalent}, 
 noted \(\keq{x}{y}\), when either
    \(x>K\) and \(y>K\), or
    \(x=y\) are integers, or
    \(x\) and \(y\) are non-integers and they have the same integral part. 
    Two~valuations~$v$ and~$v'$ are $K$-equivalent if
    $\keq{v(x)}{v'(x)}$ for all~$x\in\Clocks[]$. %We~extend
    %this to configurations and timed word with resets in the natural way.
    We say that two configurations are $K$-equivalent 
    when their valuations are, and that two timed words with reset are 
    $K$-equivalent when they have the same size and the configurations of 
    same indices in both words are $K$-equivalent.
    %by saying that two valuations 
    %\(v\) and \(u\) are $K$-equivalent if for any clock \(x_a\in\Clocks\), 
    %\(\keq{v(x_a)}{u(x_a)}\). 
  %Finally we extend this notion to timed words with resets: two words are $K$-equivalent if 
%  all the valuations \(v_i\) and \(v_i+t_i\) encountered along them are.
\end{definition}
%We call $K$-equivalent classes the zones that are $K$-equivalent. This notion
Notice that $K$-equivalence is coarser than the usual notion of
\emph{region equivalence} of~\cite{AD94}, as it aims to encode direct
indistinguishability by a guard along words, instead of
indistinguishability in the future.

We call \emph{zone-word with resets} a timed word with resets in which all valuations are replaced with
zones.
%\NM{Definition is a bit loose. If we have room I will clarify it}
%We define the satisfiability between a run and a zone run
%
A~timed word with resets~$r=((l_i,v_i)\xrightarrow{e_i}(l_{i+1},v_{i+1}))_{i\in[0,n]}$ is
\emph{compatible} with a zone
word with resets~$zr=((l_i,z_i)\xrightarrow{e_i}(l_{i+1},z_{i+1}))_{i\in[0,n]}$,
written \(r\models zr\), when $v_i\models z_i$ for all~$i$.  We call
\emph{$K$-closed word}
%\NM{Je sais pas si cette notion intervient
%  souvent, mais ``equivalent'' est g\'en\'eralement une relation
%  binaire. Je dirais ``K-closed'' par exemple. Mais \c ca risque
%  d'etre difficile de changer toutes les occurrences... }
a zone word in which all zones are $K$-equivalence classes.

\begin{lemma}
  For any timed word with reset \(r\) of a (deterministic) timed automaton~\(\ta\),
  there is a unique $K$-closed word~\(zr\) such that
  \(r \models zr\). For any timed word with resets~\(r'\)
  %timed word with reset  
  compatible with~\(zr\), \(r'\)~is accepting if, and only~if,
  \(r\)~is.
\end{lemma}
%% remarque ci-dessous arrive trop tot, on ne la comprend pas.
%The previous lemma shows the central property of $K$-equivalent-words\NM{$K$-equiv word = $K$-equiv run? Pourquoi changer de terminologie ?} for a learning method: they correspond to the finest grain to which a TA can distinguish past behaviours. 

\emph{Event recording automata}~(ERA)~\cite{AFH99} are a subclass of
TAs in which there is one clock~$x_a$ per letter~$a$ of the alphabet,
such that $x_a$ is reset exactly along $a$-transitions. We~slightly
extend them as \emph{reset-free ERAs} (RERAs), in which transitions
may or may not reset their clock:
%, handled in \cite{GJP06}, and our generalization: the \emph{reset-free} ERA (RERA). In these models, clocks are tied to letters of the alphabet. We will note \(\Clocks=\{x_a\mid a\in\Sigma\}\), and define in the same way \(\Zones\) and \(\Guards\).
%As defined just below, a RERA may only reset the clock corresponding to the letter of the transition. If the reset is always enforced, it is furthermore an ERA.
we~let $\Clocks=\{x_a\mid a\in\Sigma\}$, and $\Zones$ and $\Guards$ be shortcuts for \(\Zones[{\Clocks}]\) and \(\Guards[{\Clocks}]\) respectively.
\begin{definition}%{}{resa}
  A \emph{reset-free event recording automaton} (RERA) over \(\Sigma\)
  is a TA \(\ta=(\Loc,\initloc,\Clocks,\Trans,\Accept)\) such that
  % there is one clock per letter in the alphabet, and
  for all transitions \((\loc,a,g,r,\loc')\in\Trans\), it~holds
  \(r\in\{\{x_a\},\emptyset\}\).
  %An \emph{event recording automaton}
  %(ERA) is a RERA where the reset is furthermore constrained to always
  %be \(\{x_a\}\).
\end{definition}

\begin{example}
  Consider the timed automaton depicted in Fig.~\ref{fig:ta}.  This TA
%in Fig.~\ref{fig:ta} can be seen
is actually a RERA, by associating clock \(x\) to letter \(b\) and clock \(y\) to letter \(a\). 
  An~accepting timed word with resets of this automaton is
  \((\initloc,\initv)\xrightarrow{1.5}(\initloc,\icol{1.5\\1.5})\xrightarrow{a,\emptyset}
  (\loc_1,\icol{1.5\\1.5})\xrightarrow{b,\{x\}}(\initloc,\icol{0\\1.5})\xrightarrow{a,\emptyset}(\loc_1,\icol{0\\1.5})\xrightarrow{2}(\loc_1,\icol{2\\3.5})\xrightarrow{a,\emptyset}(\loc_2,\icol{2\\3.5})\).
  The corresponding path is
  \(\initloc\xrightarrow{e_1}\loc_1\xrightarrow{e_2}\initloc\xrightarrow{e_1}\loc_1\xrightarrow{e_3}\loc_2\).
  %Notice that we still give the name of the transition in paths, as multiple transitions could exist between the same pair of states with the same letters and resets, even in deterministic TA.
%\end{example}
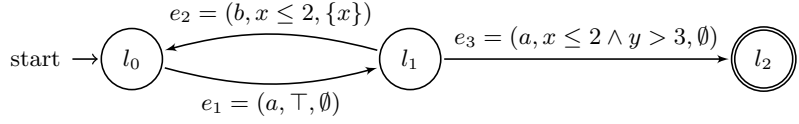
\begin{figure}[t]
  \centering
  \begin{tikzpicture}[->,shorten >=1pt,shorten <=1pt,auto,semithick, node distance = 3.7cm]
    \node[initial,state] (0) {$\initloc$};
    \node[state] (1) [right of=0] {$\loc_1$};
    \node[state,accepting] (2) [right of=1, node distance=4.7cm] {$\loc_2$};
    \path[use as bounding box] (0,-.7);
%    \everymath{\scriptstyle}
    \path (0) edge [bend right=15,-latex'] node [below]{$e_1=(a,\true,\emptyset)$}  (1)
          (1) edge [bend right=15,-latex'] node [above] {$e_2=(b,x\leq2,\{x\})$}     (0)
          (1) edge [-latex']            node {$e_3=(a,x\leq2\wedge y>3,\emptyset)$} (2);
  \end{tikzpicture}
  \caption{A simple TA}
  \label{fig:ta}
\end{figure}
%\begin{example}
\end{example}

Although closely related, ERA and RERA differ in a central way
w.r.t. our learning problem: while the resets of an~ERA can be
directly inferred from observations, in a RERA this is not directly
possible.  Thus, generalizing a learning method from ERA to RERA
requires dealing with the inference of resets---one~of the central
challenges of the learning of general deterministic~TA.

%%%%%%%%%%%%%%%%%%%%%%%%%%%%%%%%%%%%%%%%%%%%%%%%%%%%%%%%%%%%
%%%%%%%%%%%%%%%%%%%%%%%%%%%%%%%%%%%%%%%%%%%%%%%%%%%%%%%%%%%%
%%%%%%%%%%%%%%%%%%%%%%%%%%%%%%%%%%%%%%%%%%%%%%%%%%%%%%%%%%%%
%%%%%%%%%%%%%%%%%%%%%%%%%%%%%%%%%%%%%%%%%%%%%%%%%%%%%%%%%%%%
%%%%%%%%%%%%%%%%%%%%%%%%%%%%%%%%%%%%%%%%%%%%%%%%%%%%%%%%%%%%
%%%%%%%%%%%%%%%%%%%%%%%%%%%%%%%%%%%%%%%%%%%%%%%%%%%%%%%%%%%%
%%%%%%%%%%%%%%%%%%%%%%%%%%%%%%%%%%%%%%%%%%%%%%%%%%%%%%%%%%%%
%%%%%%%%%%%%%%%%%%%%%%%%%%%%%%%%%%%%%%%%%%%%%%%%%%%%%%%%%%%%
%%%%%%%%%%%%%%%%%%%%%%%%%%%%%%%%%%%%%%%%%%%%%%%%%%%%%%%%%%%%

\subsection{Timed languages}
Automata-learning techniques are based on the identification of a
candidate automaton that generalizes the \emph{observations} obtained
during the learning process. Angluin's tabular approach~\cite{Ang87a}
directly identifies a set of observations (\ie~words) having good
properties, and builds a deterministic automaton from~it. Our
contribution, as well as all the active-learning algorithms that we
are aware~of, follow a similar approach.
%
%A~key difference with timed words in this regard is that there is an
%
An important issue for extending this approach to timed words is the
infinite number of observations fitting even the simplest model, due
to time density.  We~thus have to use good abstractions to represent
classes of these words, and use these classes to direct the learning
process.
%This~was made in~\cite{GJP06}, albeit a bit informally.
A~first such extension was initiated in~\cite{GJP06}.

%Thus we define timed languages independently of any given
%model, and deploy several abstractions aiming at representing them
%efficiently. We fix an alphabet \(\Sigma\) and as we aim at learning
%RERA, we consider clocks tied to the alphabet letters.  \par

\looseness=-1
A~timed word with resets of a RERA can be seen as an element of
\((\bbR_{\geq0}\times\Sigma\times\{\true,\false\})^{*}\). A~timed word
is the projection of a timed word with resets
on~$(\bbR_{\geq0}\times\Sigma)^*$; timed words correspond to
observations of timed words with resets.
%
%We will often represent timed words with resets as runs
%of an (unknown) RERA starting from the initial valuation \(\initv\).
%\par

In order to represent infinitely many timed words with resets in a
succinct way, we define \emph{guarded words with resets}
\(w_{gr}\in(\Guards\times\Sigma\times\{\true,\false\})^{*}\), which correspond to
paths in a~RERA.
%\NM{a bit loose}
%
% \begin{remark}
%   We do not define guarded words without resets as guards have to be understood on \emph{valuations}, and no valuation can be inferred without reset information. 
% \end{remark}
%
For a timed word 
\(w_{t}\) %=(t_i,a_i)_{i\in[0,n]}\)
and a guarded word with resets 
\(w_{gr}\) %=(g_i,b_i,r_i)_{i\in[0,m]}\),
we say that \emph{\(w_{t}\) satisfies \(w_{gr}\)}, noted 
\(w_{t}\models w_{gr}\), if 
% \(n=m\) and for all \(i\in[0,n]\), \(a_i=b_i\) and the timed words with resets \(w_{tr}=(v_i\xrightarrow{(t_{i},a_{i},r_i)}v_{i+1})_{i\in[0,n]}\) constructed by adding the reset information from the guarded word with resets, verifies \(v_i+t_i\in g_i\). 
$w_t$ is a possible observation of~$w_{gr}$.
%\LH{J'ai simplifié et dé-mathématisé. Un peu trop? NM: c'est bien comme ca je pense}
%
We extend this correspondence to timed words with resets by ensuring
that the resets match.
The satisfiability relation between timed words and guarded words with
resets will be central in the rest of the paper, as it relates an
observation to the unfolding of a RERA (or~of our hypothesis).
%\NM{dire pourquoi on parle d'unfolding...}

\begin{example}
  The timed word \(w_t=(1.3,a).(0.4,b)\) satisfies
  the guarded word with reset
  \(w_{gr}=(x_b>1,a,\{x_a\}).(x_a<1,b,\emptyset)\): indeed, \(w_t\)
  and \(w_{gr}\) have the same untimed projection, and the timed word
  with resets
  \(w_{tr}=\initv\xrightarrow{1.3}\icol{1.3\\1.3}\xrightarrow{a,\{x_a\}}\icol{0\\1.3}\xrightarrow{0.4}\icol{0.4\\1.7}\)
  satisfies the guards of \(w_{gr}\).
  %: \(0+1.3\models x_b>1\) and
  %\(0+0.4\models x_a<1\).
  Notice that \(w_t \not\models
  w'_{gr}=(x_b>1,a,\emptyset).(x_a<1,b,\emptyset)\), as modifying resets
  changes the valuations that appear in the corresponding timed
  word with resets.
\end{example}

% \par Guarded words with resets have all the information necessary to encode the configurations 
% reached along a set of similar timed words, but do not do so explicitly. 
\emph{Zone words with resets} can be seen as elements \(w_z\) of \((\Zones\times\Sigma\times\{\true,\false\})^{*}.\Zones\). 
From a guarded word with resets \(w_{gr}=(g_i,a_i,r_i)_{i\in[0,n]}\) we can define the 
corresponding zone word with resets \(w_z=(z_i,a_i,r_i)_{i\in[0,n]}z_{n+1}\) with \(z_0=\{\initv\}\future\)
and \(z_{i+1}=(z_i\wedge g_i)\future\) if \(r_i=\false\) and 
\(z_{i+1}={(z_i\wedge g_i)_{[x_{a_i}\leftarrow0]}}\future\) otherwise. 

In our learning process, we will manipulate linear combinations of timed words. 
For two timed words \(w_t^1=((t^1_i,a_i))_{i\in[0,n]}\) and 
\(w_t^2=((t^2_i,a_i))_{i\in[0,n]}\) with the same untimed projection,
we define their \emph{$\lambda$-weighted sum}
%combination by a factor \(\lambda\in[0,1]\), 
%written
\(w_t^3=\lambda.w_t^1+(1-\lambda)w_t^2\), as the timed word 
\(w_t^3=((\lambda.t^1_i+(1-\lambda).t^2_i,a_i)_{i\in[0,n]})\). 
%One key property of resets is that they are linear:
Such linear combinations have the following property:
\begin{restatable}{proposition}{clockcomb}
  \label{pr:clock_comb}
  For any two timed words \(w_t^j=(t^j_i,a_i)_{i\in[0,n]}\) for \(j\in\{1,2\}\) with the same untimed projection, for any \(\lambda\in[0,1]\) and for any reset word \((r_i)_{i\in[0,n]}\), 
  all the valuations \(v_{i,r}^3\) reached along \(w_{tr}^3=((\lambda.t^1_i+(1-\lambda).t^2_i,a_i,r_i)_{i\in[0,n]})\) 
  are such that for all clocks \(x_a\in\Clocks\), 
  \(v_{i,r}^3(x_a)=\lambda.v_{i,r}^1(x_a)+(1-\lambda).v_{i,r}^2(x_a)\) for \(v_{i,r}^j\) the valuations reached along \(w_{tr}^j=(t^j_i,a_i,r_i)\).
  %\lNM{le reset word $r$ n'est pas utilis\'e...}
\end{restatable}

% This property central interest is that the result holds for any reset 
% combination \(r\), which will allow us to transfer properties from a reset combination to 
% another one while using the same set of observations. 

%% file: structure.tex
The general principle of (untimed) active-learning is to learn a model from observations acquired by membership queries and equivalence queries\iflong~\ref{fig:active_learning}\else\fi.
In membership queries, a timed word is provided to a teacher, who in return informs us about the membership of this world in the target language. 
In~an equivalence query, we~propose an hypothesis (model) to the teacher; she~either accepts it if~it~is equivalent to the model we wish to learn, 
or otherwise provides us with a counterexample, \ie, a~timed word that separate the language of the model and that of our hypothesis.
%, \ie~a misprediction of the hypothesis. 
\iflong
\begin{figure}[h]
  \centering
  \begin{tikzpicture}
    \node [draw, rectangle, minimum height = 2cm, text width= 1.5cm] (L) {Learner};
    \node [ rectangle, minimum height = 2cm, right of = L, node distance = 4.5cm] (A) {};
    \node [draw, rectangle, minimum height = 2cm, text width= 1.5cm, right of = A, node distance = 4.5cm] (T) {Teacher};

    \path ( $(L.north east)!0.22!(L.south east)$) edge [->] node [above] {Membership $w_t$?} ($(T.north west)!0.22!(T.south west)$)
          ($(L.north east)!0.72!(L.south east)$) edge [->] node [above] {Equivalence TA?} ($(T.north west)!0.72!(T.south west)$)
          ($(T.north west)!0.28!(T.south west)$) edge [->] node [below] {Observation} ($(L.north east)!0.28!(L.south east)$)
          ($(T.north west)!0.78!(T.south west)$) edge [->] node [below] {Yes/Counterexample} ($(L.north east)!0.78!(L.south east)$);
  \end{tikzpicture}
  \caption{The basic active learning framework}
  \label{fig:active_learning}
\end{figure}
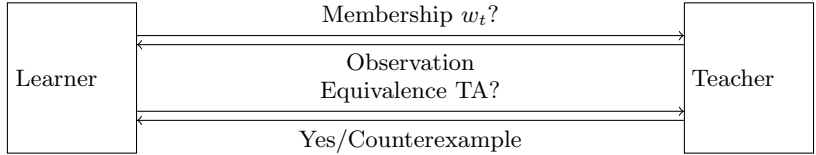
\else\fi
The set of observations is formalized as a partial function~\(\Obs\) mapping words to acceptance status (\(\ltrue\) or~\(\lfalse\)).
To~build a model, we~then want to identify a prefix-closed subset~$U$ of~\(\Dom(\Obs)\) such
that for all letters~$a$ in the alphabet and words~$u\in U$, \(u.a\in\Dom(\Obs)\) and either \(u.a\in U\), or there is another
word \(u'\in U\) having the same observed behaviour as~$u.a$.
When transferring this approach to timed words, one has to deal with
two difficulties: first, the uncountable number of possible delays
before each discrete action; second, the fact that observations do not
include clock valuations (nor clock resets), which we also have to learn.

In this section, we describe the structures used to represent and
process these timed observations acquired during the learning \emph{and} the
decisions on the built structures made based on those observations.
We generalize timed decision trees defined in~\cite{GJP06}, 
so as to encode timed words \emph{with possible resets}.
We basically use a {\em timed decision graph}, a model close to acyclic timed automata, to encode the current knowledge inferred about the model from observations,
%In~order to detect such problems, we~use a
and a \emph{timed observation graph}~(TOG) to \emph{implement}~\(\Obs\) with a step of abstraction and help decisions.
Our~data structure is centered around the notion of
\emph{observation structure} composed of a \emph{timed decision
graph}, which stores the current hypothesis (and will later be folded
into a~TA), and an \emph{observation function}, which stores current
observations.
%

%In order to fully exploit the information gained
%in \emph{all reset scenarios}, we furthermore define a \emph{timed
%observation graph}, which we will use to implement the observation
%function and detect reset combinations that can not happen according
%to the observations.
%\TJ{a garder ou plus loin ?}

%We first define timed decision graphs:
%As explained above, those graphs are actually
%as bipartite trees, with language states corresponding to guarded words
%with resets, and decision states implementing the different possible
%resets. All the information that will be encoded in this graph should
%be backed by observations. Notably, the guards introduced are
%necessary to disambiguate those observations, in the possible reset
%combinations. Formally:

\begin{definition}%{Observation structure}{}
  An \emph{observation structure} is a pair \((\tdg,\Obs)\) made of a
  {\em timed decision graph} (TDG)  and
  a partial mapping \(\Obs\) from timed words  to $\{\ltrue,\lfalse\}$.
  The TDG is a labelled bipartite graph
  \(\tdg=(\Sta,\Trans)\) with $\Sta=\Stl\uplus\Std$ where: 
  \begin{itemize}
  \item
    \(\Stl\subseteq\{s_0=(\epsilon,\{\initv\}\future)\} \cup (\Guards\times\Sigma\times\{\true,\false\})^{+}\times\Zones\)
     is a set of \emph{language states},
    made of a 
    prefix-closed finite set of guarded words with resets paired with
    zones; \(s_0\)~is the root state.
      \item \(\Std\subseteq \Stl\times\Sigma\times\Guards\)
      is a set of
    \emph{decision states} such that  
%    extending any language state \(\stl\) into decision states with a letter \(a\) realizes a partition of the clock constraints, \ie:
    for any \(\stl\in\Stl\) and \(a\in\Sigma\),
    if \(I=\{g\in\Guards\mid (\stl,a,g)\in \Std\}\) is non-empty, then \(\bigvee_{g\in I} g \equiv \true\) and for all
    $g$ and~$g'$ in~$I$,  if $g\neq g'$ then $g\wedge g'\equiv \false$;
  \item \(\Trans\subseteq \Sta\times(\Sigma\times\Guards \cup
    \{\true,\false\})\times\Sta\) is defined such that
    transitions to a decision state \(\std=(\stl,a,g)\) are of the form 
    \((\stl,a,g,\std)\) and if  \(\stl=w_{gr}.z\)
    transitions from \(\std\) are 
    \((\std,\true,(w_{gr}.(g,a,\true),{(z\wedge g)_{[x_a\leftarrow0]}}\future))\) and 
    \((\std,\false,(w_{gr}.(g,a,\false),\penalty1000\relax(z\wedge g)\future))\).
  \end{itemize} 

  The \emph{labelling} of an
  observation structure maps language states to the set of observations compatible with them:
  \[\stlabel(\stl=(w_{gr}.z))=\{\Obs(w_t)\mid w_t\in \Dom(\Obs)\wedge 
  w_t\models w_{gr} \}.\]
\end{definition}
It can be seen from this definition that TDGs are trees 
\iflong (see the proof in Appendix~\ref{app:proofs})\else\fi.
%Paths to a language state \(\sl\) are labelled with guarded words with resets.
For a guarded word \(w_{gr}\), we note \(\stinit\xrightarrow{w_{gr}}_\tdg \stl\) when
there is a path in~\(\tdg\)  from $\stinit$ to \(\stl\) labelled with \(w_{gr}\), 
and note \(w_{gr}\in\tdg\) when such a path exists.

Observation structures store both the words that have been observed (in \(\Obs\)) and the inferred guards and enforced resets (or~absence thereof) (in \(\tdg\)).
We can extend \(\Obs\) to guarded words with resets by considering them as language states and using their labels. The labels are used to carry the observation information to the TDG.

\begin{figure}[t]
  \begin{subfigure}[b]{.5\linewidth}
    \centering
  %\resizebox{1\textwidth}{!}{
  \begin{tikzpicture}[->,shorten >=1pt,shorten <=1pt,auto,semithick, node distance = 1cm,scale=0.4]
    \begin{pgfinterruptboundingbox}
      \node[so,initial] (e) {$\lfalse$};
    \end{pgfinterruptboundingbox}
    \path(0,0);
      \node[sd] (a) [below of=e] {};
      \node[so] (at) [left=1.2 of a] {$\lfalse$};
      %\node (label) [above left=0 of at] {$w_{gr}$};
      \node[sd] (ata) [below of= at] {};
      \node[so] (atat) [left of=ata] {$\pm$};
      \node[so] (ataf) [right of=ata] {$\pm$};
      \node[so] (af) [right=1.2 of a] {$\lfalse$};
      \node[sd] (afa) [below of= af] {};
      \node[so] (afat) [left of=afa] {$\pm$};
      \node[so] (afaf) [right of=afa] {$\pm$};
      \path (e) edge node {$a,\gtrue$} (a)
            (a) edge node[above left] {$\true$} (at)
            (at) edge node {$a,\gtrue$} (ata)
            (ata) edge node [below] {$\true$} (atat)
            (ata) edge node [below] {$\false$} (ataf)
            (a) edge node {$\false$} (af)
            (af) edge node {$a,\gtrue$} (afa)
            (afa) edge node [below] {$\true$} (afat)
            (afa) edge node [below] {$\false$} (afaf);
    \end{tikzpicture}
  %  }
    \subcaption{A first observation structure}
    \label{fig:intro_incons}
  \end{subfigure}
  \begin{subfigure}[b]{.5\linewidth}
    \centering
%      \resizebox{.90\textwidth}{!}{
    \begin{tikzpicture}[->,shorten >=1pt,shorten <=1pt,auto,semithick, node distance = 1.2cm,scale=0.4]
      \node[initial,so] (e) {};
      \node[so] (a) [right of=e] {};
      \node[so,accepting] (aa) [right of=a] {};
      \path (e) edge node [above, yshift=8pt] {$a,\gtrue,\{x_a\}$} (a)
            (a) edge node [below, yshift=-8pt] {$a,x_a\leq1,\emptyset$} (aa);
    \end{tikzpicture}
 %   }
    \subcaption{The RERA providing observations}
    \label{fig:intro_rera}
  \end{subfigure}
  \caption{An active-learning setting}
\end{figure}
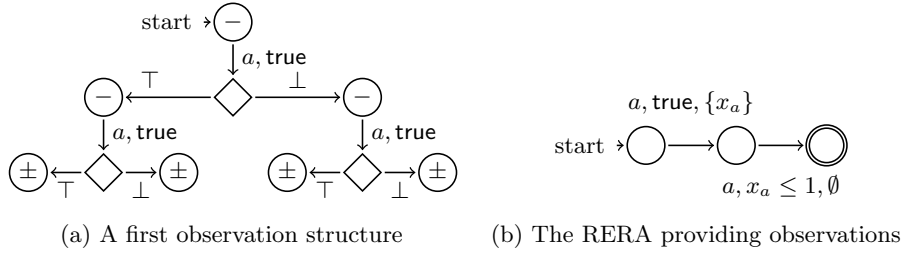
\begin{example}
  \label{ex:tdg}
  Fig.~\ref{fig:intro_incons} represents an observation structure storing some words observed from the RERA in Fig.~\ref{fig:intro_rera}. Language states are depicted as circles and decision states as diamonds. 
  Notice that in this example the leaves have labels of size $2$: they model both accepting and non-accepting observations \eg~\(((0.7,a)(0.9,a),\ltrue)\) and \(((0.7,a)(1.2,a),\lfalse)\).
\end{example}

\par We define some desired properties of information structures.
%These properties are defined more generally on subgraphs, as this will be useful to prove algorithms manipulating the observation structure. 
\begin{definition}%{Properties of an observation structure}{compl/consist}
  For an observation structure \((\tdg,\Obs)\), a subtree \(\tdg'\) of 
  \(\tdg\) rooted in \(\stl^{\tdg'}=w'_{gr}.z'\) is said:
  \begin{itemize}
    \item \emph{complete} when all observations in \(\Obs\) are taken into account, i.e.
    for any \(w_t\in \Dom(\Obs)\) such that \(w_t=w'_t.w''_t\) with \(w'_t\models w'_{gr}\) 
    there is \(\stl^{\tdg'}\xrightarrow{w''_{gr}}\stl\) such that \(w''_t\models w''_{gr}\) and for all such \(w''_{gr}\) and \(\stl\), 
    \(\Obs(w_t)\in \stlabel(\stl)\);
    \item \emph{consistent} when it separates accepting and non accepting behaviours, i.e. for any \(\stl\) in the subtree, \(|\stlabel(\stl)|=1\).
  \end{itemize}
  We say that an observation structure is complete or consistent when \(\tdg\) is.
  % An observation structure \((\tdg,\Obs)\) is said: 
  % \begin{itemize}
  %   \item \emph{complete} when for any \(w_t\in \Dom(\Obs)\) there is \(w_{gr}\in\tdg\) such 
  %   that \(w_t\models w_{gr}\) and for all such \(w_{gr}\)and \(\stl\) such that 
  %   \(\stinit\xrightarrow{w_{gr}}\stl\), \(\Obs(w_t)\in \stlabel(\stl)\).
  %   \item \emph{consistent} if for all \(w_{gr}\in\stl\), \(|\Obs(w_{gr})|=1\).
  % \end{itemize}
\end{definition}
%Intuitively a complete observation structure takes into account all the observation in \(\Obs\),
%while a consistent one correctly separates the accepting and non-accepting behaviours.

Detecting and handling inconsistencies is central to our algorithms,
as it characterizes the need to introduce new guards to split language
nodes in the timed decision graph. 
\begin{figure}[htp]
    \begin{subfigure}[b]{.49\linewidth}
      \centering
%      \resizebox{\textwidth}{!}{
    \begin{tikzpicture}[->,shorten >=1pt,shorten <=1pt,auto,semithick, node distance = .9cm]
    \begin{scope}
      %      \begin{pgfinterruptboundingbox}
      \node[so,initial] (e) {$\lfalse$};
%      \end{pgfinterruptboundingbox}
%      \path(0,0);
      \node[sd] (a) [below of=e] {};
        \node[so] (at) [below left=.09 and .9 of a] {$\lfalse$};
        \node[sd] (ata1) [below left=.15 and .8 of at] {};
        \node[so] (ata1t) [below left = .19 and .3 of ata1] {$\ltrue$};
        \node[so] (ata1f) [below right = .19 and .3 of ata1] {$\ltrue$};
        \node[sd] (ata2) [below right=.15 and .8 of at] {};
        \node[so] (ata2t) [below left = .19 and .3 of ata2] {$\lfalse$};
        \node[so] (ata2f) [below right = .19 and .3 of ata2] {$\lfalse$};
        \node[so] (af) [below right=.09 and .7 of a,color=red!80!black] {$\lfalse$};
        \node[sd] (afa) [below right=.15 and .5 of af,color=red!80!black] {};
        \node[so] (afat) [below left = .19 and .3 of afa,color=red!80!black] {$\pm$};
        \node[so] (afaf) [below right = .19 and .3 of afa,color=red!80!black] {$\pm$};

        %%   \node[sd] (a) [below of=e] {};
        %% \node[so] (at) [below left=.05 and 1.2 of a] {$\lfalse$};
        %% \node[sd] (ata1) [below left=.1 and 1 of at] {};
        %% \node[so] (ata1t) [below left = .15 and .6 of ata1] {$\ltrue$};
        %% \node[so] (ata1f) [below right = .15 and .6 of ata1] {$\ltrue$};
        %% \node[sd] (ata2) [below right=.1 and 1 of at] {};
        %% \node[so] (ata2t) [below left = .15 and .6 of ata2] {$\lfalse$};
        %% \node[so] (ata2f) [below right = .15 and .6 of ata2] {$\lfalse$};
        %% \node[so] (af) [below right=.05 and 1.2 of a,color=red!80!black] {$\lfalse$};
        %% \node[sd] (afa) [below right=.1 and 1 of af,color=red!80!black] {};
        %% \node[so] (afat) [below left = .15 and .6 of afa,color=red!80!black] {};
        %% \node at (afat) {$\ltrue/\lfalse$};
        %% \node[so] (afaf) [below right = .15 and .6 of afa,color=red!80!black] {$\scriptstyle \ltrue/\lfalse$};
    \end{scope}
        \path (e) edge node {$a,\gtrue$} (a)
              (a) edge node[above left] {$\true$} (at)
              (at) edge node[above left=-.1mm and -2mm] {$a,x_a\leq 1$} (ata1)
              (ata1) edge node[above left] {$\true$} (ata1t)
              (ata1) edge node {$\false$} (ata1f)
              (at) edge node[above right=-.1mm and -2mm] {$a, x_a>1$} (ata2)
              (ata2) edge node[above left] {$\true$} (ata2t)
              (ata2) edge node {$\false$} (ata2f)
              (a) edge [color=red!80!black] node {$\false$} (af)
              (af) edge [color=red!80!black] node {$a,\gtrue$} (afa)
              (afa) edge [color=red!80!black] node[above left] {$\true$} (afat)
              (afa) edge [color=red!80!black] node {$\false$} (afaf);
      \end{tikzpicture}
%      }
    \caption{After handling the inconsistency}
    \label{fig:intro_inval}
  \end{subfigure}
  \begin{subfigure}[b]{.49\linewidth}
  \centering
  \resizebox{\textwidth}{!}{
  \begin{tikzpicture}[->,shorten >=1pt,shorten <=1pt,auto,semithick, node distance = .9cm]
%    \begin{pgfinterruptboundingbox}
    \node[so,initial] (e) {$\lfalse$};
%    \end{pgfinterruptboundingbox}
%    \path(0,0);
    \node[sd] (a) [below of=e] {};
    \node[so] (at) [below left=.09 and .9 of a] {$\lfalse$};
    \node[sd] (ata1) [below left=.15 and .8 of at] {};
    \node[so] (ata1t) [below left = .19 and .3 of ata1] {$\ltrue$};
    \node[so] (ata1f) [below right = .19 and .3 of ata1] {$\ltrue$};
    \node[sd] (ata2) [below right=.15 and .8 of at] {};
    \node[so] (ata2t) [below left = .19 and .3 of ata2] {$\lfalse$};
    \node[so] (ata2f) [below right = .19 and .3 of ata2] {$\lfalse$};
    \node[so] (af) [below right=.09 and .7 of a] {$\lfalse$};
    \node[sd] (afa) [below right=.15 and .5 of af] {};
    \node[so] (afat) [below left = .19 and .3 of afa] {$\pm$};
    \node[so] (afaf) [below right = .19 and .3 of afa] {$\pm$};
    %% \node[sd] (a) [below of=e] {};
    %% \node[so] (at) [below left=.05 and 1.2 of a] {$\lfalse$};
    %% \node[sd] (ata1) [below left=.1 and 1 of at] {};
    %% \node[so] (ata1t) [below left = .15 and .6 of ata1] {$\ltrue$};
    %% \node[so] (ata1f) [below right = .15 and .6 of ata1] {$\ltrue$};
    %% \node[sd] (ata2) [below right=.1 and 1 of at] {};
    %% \node[so] (ata2t) [below left = .15 and .6 of ata2] {$\lfalse$};
    %% \node[so] (ata2f) [below right = .15 and .6 of ata2] {$\lfalse$};
    %% \node[so] (af) [below right=.05 and 1.2 of a] {$\lfalse$};
    %% \node[sd] (afa) [below right=.1 and 1 of af] {};
    %% \node[so] (afat) [below left = .15 and .6 of afa] {$\ltrue/\lfalse$};
    %% \node[so] (afaf) [below right = .15 and .6 of afa] {$\ltrue/\lfalse$};
      \path (e) edge node {$a,(0,1)$} (a)
            (a) edge node[above left] {$\true$} (at)
            (at) edge node[above left] {$a,(0,1)$} (ata1)
            (at) edge node {$a,(1,2)$} (ata2)
            (ata1) edge node[above left] {$\true$} (ata1t)
            (ata1) edge node {$\false$} (ata1f)
            (ata2) edge node[above left] {$\true$} (ata2t)
            (ata2) edge node {$\false$} (ata2f)
            (a) edge node {$\false$} (af)
            (af) edge node {$a,(1,2)$} (afa)
            (afa) edge node[above left] {$\true$} (afat)
            (afa) edge node {$\false$} (afaf);
  \end{tikzpicture}
  }
  \caption{A timed observation graph}
  \label{fig:intro_tog}
  \end{subfigure}
\end{figure}
\begin{example}
  The leaves of the TDG in Fig.~\ref{fig:intro_incons} are inconsistent. The inconsistency can be resolved for the left branch by splitting the transition, as made in Fig.~\ref{fig:intro_inval}. This leaves a label of size two in the right branch, but there exists no guard that can separate the observations.
\end{example}

%% The observation function, even if considered as a partial function,
%% must be represented in memory by a fitting data structure. A~priori,
%% a~simple table, or a more efficient equivalent structure (\eg~hash table) is
%% sufficient to represent such data.
%% In~the following, we~argue that
%% using a more complex data structure allows to detect on-the-fly that
%% some resets are either impossible, or mandatory.  We~present a
%% bipartite tree similar to the TDG, but aiming to store the
%% undistinguishable tube around each observation (\ie~the
%% $K$-equivalent-words with resets), and detect on-the-fly, as observations
%% are added, that two observations sharing the same K-closed word do not agree on acceptance.

%The following definition is "loose", as we only define a model with
%few constraints on~it, and will later use graphs that \emph{implement}
%a given observation function. For~example, this definition uses zones
%instead of $K$-equivalent classes, and do not exactly specify the set
%of transitions that should be present.\LH{On peut s'en passer?}

We now define \emph{timed observation graphs}, a structure used to encode the observation function $\Obs$ efficiently and abstractly.
More precisely, it represents the undistinguishable tube around each observation (\ie~the
$K$-closed-words with resets), and allows to detect on-the-fly when two observations sharing the same K-closed word do not agree on acceptance
and when reset combinations cannot happen.

\begin{definition}%{Timed observation graph}{}
  A \emph{timed observation graph} (TOG) is a TDG where all guards and zones correspond to K-equivalence classes, language states are called \emph{observation states} \(\sto\in\Sto\) and transitions from decision to observation states do not use the future operator, \ie~for \((\std=(w.z),\true,\sto)\in \Trans\), \(\sto=w.(g,a,\true).g_{[x_a\leftarrow 0]}\) and same for \(\false\). We add a labelling \(\olabel\colon \Sto\rightarrow\calP(\{\ltrue,\lfalse\})\) for observation states and 
  \(\words\colon \Sto\rightarrow\calP((\bbR_{\geq0}.\Sigma)^{*})\) 
  a function associating to each observation state a set of observations that it represents. %\TJ{Le labelling $\olabel$ est il le meme que $label$ ?}
  For two observation states~$\sto$ and~$\sto'$, we note \(\sto\xrightarrow{w_{zr}}\sto'\) if there is a path 
  from \(\sto\) to~\(\sto'\) and there exists a zone word \(w.z\) such that 
  \(\sto=w.z\) and \(\sto'=w.w_{zr}\).
\end{definition}
As for TDGs, TOGs are trees% 
\iflong\ (see the proof in appendix~\ref{app:proofs})\else\fi.
Timed observation graphs will allow to detect impossible combinations of resets denoted by labels of observation states
of cardinality larger than one.
This is ensured by an encoding of~$\Obs$ into the~TOG, in~a way
defined as follows:
\begin{definition}%{Explicit implementation of an observation function}{}
  A timed observation graph \(\tog\) is said to \emph{implement}
  an observation function~\(\Obs\) when the following two
  conditions are fulfilled:
  \begin{description}
  % \item[Precision:] all zones appearing in \(\tog\) are
  %   $K$-equivalence classes;
  \item[Correspondence:] all observations are encoded in the TOG, i.e. for
    all \(w_t\in \Dom(\Obs)\), for any \(w_{tr}\) compatible with
    \(w_t\), there is a path \(\steps\xrightarrow{w_{zr}}\sto=w_{zr}\)
    in \(\tog\) such that \(w_{tr}\models w_{zr}\), \(w_t\in
    \words(\sto)\) and \(\Obs(w_t)\in \olabel(\sto)\);
  \item[Coverage:] all observation states are covered by
    \(\Dom(\Obs)\), i.e. for any \(\sto=(w_{zr})\in \Sto\),
    \(\words(\sto)\neq\emptyset\) and for any \(w_t\in \words(\sto)\),
    \(w_t\in \Dom(\Obs)\), \(w_t\models w_{zr}\) and \(\Obs(w_t)\in
    \olabel(\sto)\).
  \end{description}
\end{definition}

\begin{example}
  \label{ex:tog}
  The TOG in Fig.~\ref{fig:intro_tog} corresponds to the observation structure displayed in our previous examples. Notice that it has a label of size two on the leafs of the right branch.
\end{example}

The pruning of the timed decision graph  relies on \emph{invalidity} of words and states, our key contribution to the active learning framework for timed automata.
%\TJ{On parle du pruning de TDG mais l'invalidite est definie sur le TOG, il faudrait faire le lien} 
It allows to characterize reset combinations that are impossible for a given $K$-closed word.
This complements inconsistency and allows to prune resets and schedule guards to be
added when resets are not tied to observations.
\begin{definition}%{Invalidity}{}
  A \emph{$K$-closed} word with reset
  \(w_{zr}=(z_i,a_i,r_i)_{i\in[0,n]}z\) is \emph{invalid} with respect
  to an observation graph \(\tog\) if one of the following conditions holds:
%  \begin{itemize}
% \item
  \(|\olabel(w_{zr})|=2\), or
%\item
a prefix of \(w_{zr}\) is invalid w.r.t. $\tog$, or
% \item
there exists \(z_{n+1},a_{n+1}\)  such that
  both \((z_i,a_i,r_i)_{i\in[0,n]}.(z_{n+1},a_{n+1},\true){z_{n+1}}_{[a\leftarrow0]}\)
  and \((z_i,a_i,r_i)_{i\in[0,n]}.(z_{n+1},\penalty1000\relax a_{n+1},\false)z_{n+1}\)
  are  invalid w.r.t.~$\tog$.
%  \end{itemize}

  A zone word with reset (or a guarded word with reset) is invalid if
  it models an invalid $K$-closed word with reset.
\end{definition}

Invalid guarded words with resets encode behaviours that can not correspond to any
model, and thus should be pruned in the TDG: 
\begin{restatable}{proposition}{prsound}%{}
  \label{pr:sound}
  If a timed observation graph \(\tog\) has an invalid observation state  \(\sto=w_{zr}\), there is no TA model having execution  \(w_{zr}\). 
%% For a timed observation graph \(\tog\) implementing an observation
%% function \(\Obs\), if an observation state \(\sto=w_{zr}\)
%% has a label
%% of cardinality two, then no timed automaton in which \(w_{zr}\) is in
%% the set of executions can be a model for the observations in \(\Obs\).
\end{restatable}
Situations may arise where a
guarded word with reset is not invalid but all its successors by a
given action~are; an example is presented below. In such
situations, two~different $K$-closed words with resets make the
successors invalid, and a guard has to be added.
%to avoid modelling invalid behaviours.

\begin{example}
\label{ex:valid}
Consider the partial set of observations
\(\{((1.7,a)(1,a),\ltrue),\ \penalty1000\relax((1.7,a)\penalty1000\relax(1.1,a),\lfalse),\ 
\penalty1000\relax
((2.9,a)(1.1,a),\lfalse),\ \penalty1000\relax
((2.7,a)(1.1,a),\ltrue)\}\) over the alphabet
\({\Sigma=\{a\}}\). The corresponding partial timed observation graph
\(\tog\) is displayed in Fig.~\ref{fig:valid_tog}\footnote{In order to
  avoid overloading the explanation, we call the observation and graph
  partial because we do not mention some of the observations that would
  be necessary to have the implementation property.}.
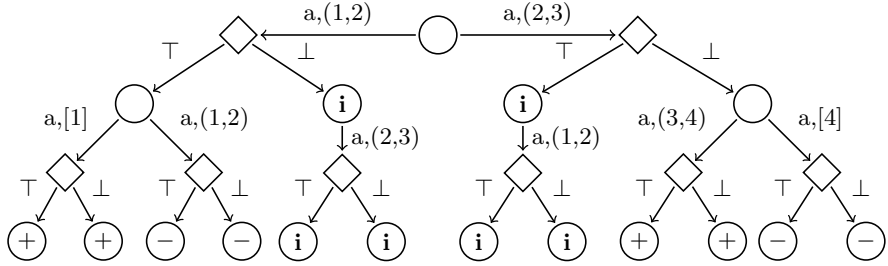
\begin{figure}[htp]
  \centering
  \newlength\picw
  \setlength\picw{9.7cm}
  %\resizebox{.90\textwidth}{!}{
  \begin{tikzpicture}[->,shorten >=1pt,shorten <=1pt,auto,semithick, node distance = 1.2cm,scale=0.4]
    %\tikzstyle{every so}=[fill=blue!30!white,text=black]
    \node[so] (e) {};
    \node[sd] (12) [left= .22*\picw of e] {};
    \node[so] (12t) [ below left= 0.6 and .11*\picw of 12] {};
    \node[sd] (12t1) [ below left=0.6 and .0625*\picw of 12t] {};
    \node[so] (12t1t) [ below left=0.6 and .02*\picw of 12t1] {$\ltrue$};
    \node[so] (12t1f) [ below right=0.6 and .02*\picw of 12t1] {$\ltrue$};
    \node[sd] (12t12) [ below right= 0.6 and .0625*\picw of 12t] {};
    \node[so] (12t12t) [ below left=0.6 and .02*\picw of 12t12] {$\lfalse$};
    \node[so] (12t12f) [ below right=0.6 and .02*\picw of 12t12] {$\lfalse$};
    \node[so] (12f) [ below right= 0.6 and .11*\picw of 12] {\linv};
    \node[sd] (12f23) [below =0.4 of 12f] {};
    \node[so] (12f23t) [ below left= 0.6 and .028*\picw of 12f23] {\linv};
    \node[so] (12f23f) [ below right=0.6 and .028*\picw of 12f23] {\linv};
    \node[sd] (23) [right= .22*\picw of e] {};
    \node[so] (23t) [ below left= 0.6 and .125*\picw of 23] {\linv};
    \node[sd] (23t12) [below=0.4 of 23t] {};
    \node[so] (23t12t) [ below left=0.6 and .028*\picw of 23t12] {\linv};
    \node[so] (23t12f) [ below right=0.6 and .028*\picw of 23t12] {\linv};
    \node[so] (23f) [ below right=0.6 and .125*\picw of 23] {};
    \node[sd] (23f34) [ below left=0.6 and .0625*\picw of 23f] {};
    \node[so] (23f34t) [ below left=0.6 and .028*\picw of 23f34] {$\ltrue$};
    \node[so] (23f34f) [ below right=0.6 and .028*\picw of 23f34] {$\ltrue$};
    \node[sd] (23f4) [ below right=0.6 and .0625*\picw of 23f] {};
    \node[so] (23f4t) [ below left=0.6 and .028*\picw of 23f4] {$\lfalse$};
    \node[so] (23f4f) [ below right=0.6 and .028*\picw of 23f4] {$\lfalse$};
    
    \path (e) edge [->] node [above] {a,(1,2)} (12)
          (12) edge [->] node[above left] {$\true$} (12t)
          (12t) edge [->] node[above left] {a,[1]} (12t1)
          (12t1) edge [->] node[above left] {$\true$} (12t1t)
          (12t1) edge [->] node {$\false$} (12t1f)
          (12t) edge [->] node {a,(1,2)} (12t12)
          (12t12) edge [->] node[above left] {$\true$} (12t12t)
          (12t12) edge [->] node {$\false$} (12t12f)
          (12) edge [->] node {$\false$} (12f)
          (12f) edge [->] node {a,(2,3)} (12f23)
          (12f23) edge [->] node[above left] {$\true$} (12f23t)
          (12f23) edge [->] node {$\false$} (12f23f)
          (e) edge [->] node [above] {a,(2,3)} (23)
          (23) edge [->] node[above left] {$\true$} (23t)
          (23t) edge [->] node {a,(1,2)} (23t12)
          (23t12) edge [->] node[above left] {$\true$} (23t12t)
          (23t12) edge [->] node {$\false$} (23t12f)
          (23) edge [->] node {$\false$} (23f)
          (23f) edge [->] node[above left] {a,(3,4)} (23f34)
          (23f34) edge [->] node[above left] {$\true$} (23f34t)
          (23f34) edge [->] node {$\false$} (23f34f)
          (23f) edge [->] node {a,[4]} (23f4)
          (23f4) edge [->] node[above left] {$\true$} (23f4t)
          (23f4) edge [->] node {$\false$} (23f4f);
  \end{tikzpicture}
  %}
  \caption{A (partial) timed observation graph with some invalid nodes.}
  \label{fig:valid_tog}
\end{figure}
We~do not represent the actual $K$-equivalent classes on the graph so as to keep the figure as simple as possible. It can be seen that both resetting and not resetting the clock after the first action may sometimes lead to an invalidity. Hence, taking these observations into account in a timed decision graph with a \(\true\) guard on this transition leads to pruning both successors of a decision tree.  

% (in order) the trees depicted in Fig.~\ref{fig:valid_tdgs}. 
% The TDG in Fig.~\ref{sfig:valid_a} corresponds to the observation 
% \(((1.7,a)(1,a),\ltrue)\) only, while Fig.~\ref{sfig:valid_b} takes into account the second observation. 
% Notice that the edge corresponding to no reset has been pruned, as it corresponds to an inconsistency, while a guard has been introduced after the reset. Adding the third observation does not modify the TDG. Finally, 
% adding the last observation requires to prune the remaining outgoing edge 
% of the first decision state, because \((a,\gtrue,\true)\) models an 
% invalidity. 
This is problematic, as a decision state should always have 
successors. Hence it is necessary to introduce a guard to distinguish the 
different invalidities. 
\end{example}

%% file: updates.tex
We define the algorithms used to update the previously defined data structures. 
The general idea is to add observations while preserving the good
properties of the data structures, which requires  detecting
inconsistencies and invalidities on-the-fly, and resolving them by adding
new guards.

The algorithms in Sec. \ref{sub:new_obs} handle new observations while keeping most of the good properties of the structures, except for consistency. When inconsistencies arise, calls are scheduled to the algorithms proposed in \ref{sub:incons}. Sec. \ref{sub:inval} deals with a similar but different problem arising from different invalidities meeting each others. Finally an algorithm to rebuild (parts of) the structure using the informations gathered using the previous section algorithms is described in Sec. \ref{sub:rebuild}.
\subsection{Adding a new observation}
\label{sub:new_obs}
%%% deja dit juste au dessus...:
%In this section, we propose algorithms to extend an observation
%structure with new observations and detect inconsistencies and
%invalidities.
%
In~essence, our algorithms propagate new words in the TDG~\(\tdg\),
using satisfiability between guarded words with resets and timed words
to guide the descent in the tree.
When new states have to be created, membership
queries are launched to get a label for them.
All of this is complemented by a similar work on the TOG~\(\tog\), in
order to take into account all the new observations. The~main
difference between the two algorithms is that in the~TDG, labels of
size~$2$ are detected and left for a future handling as the procedure
to identify guards is potentially heavy, while in the~TOG, invalidity
leads to immediate pruning in order to limit the size of the structures.

We use the functions \(\findpath_{\tdg}\)
\iflong (Algorithm~\ref{findpath_tdg})\else\fi and \(\findpath_{\tog}\)
\iflong (Algorithm~\ref{findpath_obs})\else\fi to propagate new observations in the
existing structures. Subsequent creation of new nodes is made with the
functions \(\addword_{\tdg}\) \iflong (Algorithm~\ref{addword_tdg}) \else\fi 
and \(\addword_{\tog}\) \iflong(Algorithm~\ref{addword_tog})\else\fi. 
Membership queries and the resulting function calls are handled by the 
\(\request\)
function \iflong(Algorithm~\ref{request})\else\fi, and the effective pruning is made
in \(\searchprune\) \iflong(Algorithm~\ref{searchprune})\else\fi. %All these are described in Appendix~\ref{appendix-algos}.

The \(\findpath_{\tdg/\tog}\) algorithms execute the descent through the existing
structures, while the \(\addword_{\tdg/\tog}\) ones extend the structures, and make calls to
\request. The latter algorithm first checks if a fitting observation already
exists before making a membership query if necessary.  The \searchprune
procedure follows the lines of the definition of invalidity and finds
the root of the invalid subtree before pruning~it.

% We use a pair of algorithms called \(\findpath\), applied to both the TDG \(\tdg\) and the TOG \(\tog\) implementing the observation function. These functions are complemented by a pair of \(\addword\) function that create new states and transitions. Interestingly, it is possible to both update the data structures, by adding states and labels as needed, and detect inconsistencies and invalidities on the fly, by the means of labels. Using this, we use a function \(\searchprune\) to prune exactly the invalid subtrees. When new observations are required, we use the \(\request\) function, that first checks if such information is available, and if necessary makes a new membership query. For each new observation, a new call to \(\findpath_{\tog}\) is made keep the TOG complete. 

% \par The core of the algorithms consist on a dive into the trees, calling recursively as mush as needed, allowing a great deal of parallelization. 

%We argue that \(\findpath_\tdg\) and \(\addword_\tdg\) are sound in the following way:
The following three statements express soundness of our algorithms. 
They ensure that the good properties of the structures are invariant by the call to the \(\findpath\) algorithms.
Property \ref{pr:findpath_tdg_sound} states that \(\findpath_{\tdg}\) keeps the good properties of \(\tdg\), except consistency, that is handled by in later.
Property \ref{pr:findpath_tog_sound} does the same for \(\findpath_{\tog}\) and \(\tog\), while 
property \ref{pr:searchprune_sound} ensures that the calls made to \(\searchprune\) during the execution of the \(\findpath\) algorithm prunes exactly the invalid words.
\begin{restatable}{proposition}{fptdgsound}%{Soundness of \(\findpath_\tdg\)}
    \label{pr:findpath_tdg_sound}
  Starting from a complete observation structure \((\tdg,\Obs)\) such
  that \(|\Obs(w_{gr})|\geq1\) for all \(w_{gr}\in\tdg\), and a new
  word~\(w_t\) associated with an observation \(o\), a~call to
  \(\findpath_\tdg(w_t,o,\epsilon,\initv,s_0)\) terminates and
  modifies the observation structure in such a way that
  it is complete,
  \(w_t\in \Dom(\Obs)\),
  \({\Obs(w_t)=o}\),  and  \({|\Obs(w_{gr})|\geq1}\)
  for all \(w_{gr}\in\tdg\).
\end{restatable}

%The soundness of the counterparts of these algorithms for the
%observation graph is stated below:
\begin{restatable}{proposition}{fptogsound}%{Soundness of \(\findpath_{\tog}\)}{}
    \label{pr:findpath_tog_sound}
  Starting from a timed observation graph \(\tog\) implementing an
  observation function \(\Obs\), and a new timed word \(w_t\)
  associated with the observation~\(o\), a~call to
  \(\findpath_{\tog}(w_t,o,\epsilon,\epsilon,\initv,\steps)\)
  terminates and modifies the timed observation graph in such a way that it
  implements the valid part of \(\Obs\) extended to~\(w_t\).
\end{restatable}

%We claim that the \searchprune algorithm is also sound.

\begin{restatable}{proposition}{spsound}%{Soundness of pruning}{}
    \label{pr:searchprune_sound}
  Starting from an observation structure \((\tdg,\Obs)\) where
  \(\Obs\) is implemented by \(\tog\) and no invalid states can be
  reached in \(\tdg\), calling \(\findpath_\tdg\) or
  \(\findpath_{\tog}\) modifies \(\tog\) and~\(\tdg\) in such a way that no
  invalid states can be reached in~\(\tdg\). Furthermore, no~valid words
  are made unreachable.
\end{restatable}

\subsection{Dealing with inconsistency}
\label{sub:incons}
An inconsistency arises when a language state of the TDG contains both
accepting and non-accepting observations.  It~means that a guard
must be added somewhere in the structure in order
to distinguish between these observations.  

%For this, the boundary between the accepting and non-accepting behaviours
%must be identified. This identification is the topic of this section.

For this we search for a pair of \emph{adjacent} words, which
intuitively identify the boundary between accepting and non-accepting
behaviours.
%This allows us to find
We~then build a finite set of \emph{differences} between
adjacent words, each of which corresponds to a possible guard.  This procedure
is described in the \(\adjpair\) algorithm.

\smallskip
We use $K$-equivalence to define the notion of adjacency. Intuitively
adjacent words have the same projection on actions and
resets, and their valuations either are $K$-equivalent, or they materialize a
boundary between the accepted and non-accepted words.
\begin{definition}%{Adjacency}{}
  For two timed words with resets \(w_{tr}=(v_i\xrightarrow{t_i,a_i,r_i}v_{i+1})_{i\in[0,n]}\)
  and \(w'_{tr}=(v'_i\xrightarrow{t'_i,a_i,r_i}v'_{i+1})_{i\in[0,n]}\), we say that 
  \emph{\(w_{tr}\) is adjacent to \(w'_{tr}\)} when for all \(i\in[0,n]\) and 
  \(x_a\in\Clocks\): 
  \begin{itemize}
  \item if \(v_i(x_a)+t_i\in\bbN\) then
    \(|(v_i(x_a)+t_i) - (v'_i(x_a)+t'_i) | <1\), 
    \item otherwise, \(\keq{v_i(x_a)+t_i}{v'_i(x_a)+t'_i}\).
  \end{itemize}
\end{definition}
% \begin{remark}\LH{Probably not worth the space}
%   There are two differences between the way we write the adjacency and the way it is described 
%   in \cite{GJP06}. The first one is that we use valuations plus delays, because of our definition 
%   of valuations along a timed word with resets, and the second one is that we say that a 
%   word is adjacent to another one to highlight the asymmetry of the definition. 
% \end{remark}

Notice that adjacency is not a symmetric relation.  We will sometimes
abuse the notations and say
%\NM{Je ne sais pas ce que ca demande de
%  changer, mais je pense qu'il vaudrait mieux eviter...}
that a pair
\(w,w'\) is adjacent to mean that \(w\) is adjacent to~\(w'\).  
%When a word is adjacent to another~one, we~can use this property
We~use adjacency to identify \emph{differences} between the words as
possible new guards that resolve the inconsistency.
\begin{definition}%{Difference}{}
  The \emph{difference} between two words \(w_{tr}=(v_i\xrightarrow{t_i,a_i,r_i}v_{i+1})_{i\in[0,n]}\) adjacent to \(w'_{tr}=(v'_i\xrightarrow{t'_i,a_i,r_i}v'_{i+1})_{i\in[0,n]}\), noted 
  \(\diff(w_{tr},w'_{tr})\) is the set of quadruples defined as: if for a clock $x$, \(v_i(x)+t_i=k\in\bbN\), then if \(v'_i(x_a)+t'_i<k\), \((i,x,k,\geq)\in\diff(w_{tr},w'_{tr})\) and if \(v'_i(x_a)+t'_i>k\), \((i,x,k,\leq)\in\diff(w_{tr},w'_{tr})\).
  % \[
  %   \inference[$(\geq)$]{k\in\bbN,\ v_i(x_a)+t_i=k,\ v'_i(x_a)+t'_i<k}
  %             {(i,x_a,k,\geq)\in\diff(w_{tr},w'_{tr})}
  %             \] 
  %             \[
  %   \inference[$(\leq)$]{k\in\bbN,\ v_i(x_a)+t_i=k,\ v'_i(x_a)+t'_i>k}
  %   {(i,x_a,k,\leq)\in\diff(w_{tr},w'_{tr})}
  % \] 
  % for \(v_i\) (resp. \(v'_i(x_a)\)) the valuations encountered along 
  % \(w_{tr}\) (resp. \(w'_{tr}\)).
\end{definition}
Using these definitions, we can derive from two adjacent words a set of
candidates to make a new guard.
%  but we want to further restrict this set. For this we define the 
% \emph{neighbourhood} of a word as the set \(\neighbour(w_{tr})\) of words adjacent to 
% \(w_{tr}\). We then want to find a \(w'_{tr}\) in this neighbourhood such that 
% the difference between the two words is minimal.
% \begin{definition}%{Critical pair}{}
%   A \emph{critical pair} of timed words with resets is a pair \(w_{tr},w'_{tr}\) of 
%   adjacent words that correspond to different observations and such that 
%   \(\diff(w_{tr},w'_{tr})\) is minimal for the words in \(\neighbour(w_{tr})\) that 
%   lead to different observations.
% \end{definition} 
%\subsubsection{Finding a critical pair}
\adjpair makes membership queries on linear combinations of the two initial observations to perform a binary search until the clock values of the pair have less than $1$ time unit of distance.  
Then it forces every non-$K$-equivalent pair of clock values to
have one of its elements be an integer with more linear combinations. 
Finally, in order to ensure
that only one of the two words have such integer distinctions, it
compares them with their mean. This gives an adjacent pair.

\begin{proposition}%{Local correctness and complexity}{}
  The \adjpair algorithm constructs an adjacent pair using at most 
  \(\calO(m|\Sigma|log(K))\) membership queries. %+3^{m|\Sigma|} for crit pairs
\end{proposition}
\begin{proof}
  We refer the reader to the proof of Theorem 5.8 in~\cite{GJP06}.
\end{proof}
% adjacent pair returned by \adjpair is
%stored to later generate a guard from its
%differences.

\subsection{Dealing with invalidity}
\label{sub:inval}
A label of size two in the TOG indicates an invalidity. 
%An invalidity corresponds to a label of size two in the TOG. 
It~points to a combination of resets being impossible combined with
those precise observations.  Invalidity is simply dealt with by
pruning the invalid parts of the TDG and TOG. But a challenge can
arise, as explained in Example~\ref{ex:valid}: sometimes \emph{all}
successors of a decision state of the TDG following a \emph{valid}
language state are pruned, due to invalidities. In~this case, a~guard
must be introduced to separate the different invalidities and allow to
rebuild the graph accordingly. As~for inconsistencies, it~is important
to introduce guards that model as closely as possible the changes in
behaviours of the observation.

For this purpose, we~again use a binary search, but this time
manipulating a pair of \emph{sets} of words. Furthermore, as the
invalidities are often detected by the precise combination of
\emph{fractional values}, the~delays in the words are only modified by
\emph{integer values}. For two timed words
\(w_t^i=(t_j^i,a_j)_{j\in[1,n_i]}\) with \(n_1\leq n_2\), we define
the operator \(w_t^1\op
w_t^2=(\floor{t_j^1}+\fracpart{t_j^2},a_j)_{j\in[1,n_1]}.
(t_j^2,a_j)_{j\in[n_1+1,n_2]}\)
to describe the operation used in the algorithm
(where $\floor t$ and $\fracpart t$ respectively represent
the integral and fractional parts of~$t$).

Of course, it is impossible to obtain a good precision while keeping
all fractional values: clock values can not be modified to become
integers. For this reason our algorithm only identifies a set of
integer constants separating two behaviours, but does not find which
behaviour the constants belong~to.
This means that we have to wait for a counterexample from an
equivalence query to correct %us if the wrong guess is made.
the possible wrong guesses we made.

Procedure \invalguard \iflong is described in Algorithm~\ref{invalguard}.
It~\else\fi outputs a \emph{validity guard} \((\stl,a,g,x,\mathord\sim,k)\)
where $a\in\Sigma$, $g$ is a guard, $x$ a clock, \(k\in\bbN\) and
\(\mathord\sim\in\{\mathord<,\mathord\leq,\top\}\). Such validity
guard states that in the language state~\(\stl\), after
playing~$a$ with guard~$g$, adding $x\sim k$ to the guard separates the two causes
of invalidity.
We use \(\top\) to denote that both strict and large
inequalities could fit the current observations.
The \invalguard algorithm conducts a binary search between two sets of
timed words, while keeping the fractional part of the clock values
unchanged thanks to the \(\op\) operator, while the $K$-closed sets
corresponding to the sets of words do not touch each other.

\begin{proposition}
  Algorithm \invalguard terminates after \(\calO(m(|W_1|+|W_2|)\cdot |\Sigma|\cdot \log(K))\) membership queries,
%  with $m$ the greatest length of a word in \(W_1\cup W_2\).
  where $m$ is the size of a largest word in \(W_1\cup W_2\).
\end{proposition}
\begin{proof}
  The proof uses the same arguments as the one of \adjpair. 
\end{proof}

\subsection{Rebuilding the graph}
\label{sub:rebuild}
To rebuild a subtree is to introduce new guards using adjacent pairs
and validity guards only when necessary, and re-propagate the
informations in the new guarded words with resets they satisfy. We use
Algorithm \(\rebuild\) for this.
From an adjacent pair, we extract \emph{consistency guards}, which will
be used to reconstruct a decision graph that is consistent with respect to the
adjacent pair.
\begin{definition}%{consistency guard}{}
For an adjacent pair \(w_{tr},w'_{tr}\), clock constraint \(x_a\leq
k\) is a \emph{consistency guard} at depth~$i$ if
\((i,x_a,k,\prec)\in\diff(w_{tr},w'_{tr})\) and there is no
\((j,x_a,l,\prec')\in\diff(w_{tr},w'_{tr})\) such that \(j<i\) or
\(j=i\) and \(l<k\).
\end{definition}
The consistency guards are taken on the first difference,
so as to ensure
that they can not be overwritten later (there are no guards that can
separate the pair before the guard), and
%for constants as low as possible
to avoid large
%maximal
constants as much as possible.

Notice that we can not always infer a unique guard from an adjacent pair, as multiple clocks can be different at the same time.
Intuitively, \rebuild only introduces guards "when needed", which is formalized by the following \emph{well-guardedness} property. 
\begin{definition}%{Well guardedness}{}
  \looseness=-1
  A timed decision graph is said \emph{well guarded} if, for all
  transitions \((\stl,a,g,\std)\in\Trans_{\Sigma}\) and  all
  constraints \(x_b\prec k\) in \(g\), 
  either there is \(w_{tr}\) adjacent to \(w'_{tr}\) such that both pass by 
  \(\stl\) and \(x_b\prec k\) is a consistency guard for the pair at this depth
  \emph{or} \((\stl,a,g',x_b,\sim,k)\) is a validity guard with
  $g\subset g'$ and
  $\prec$ is either $\sim$ or~$\neg\sim$.
  %either \(\mathord\prec = \mathord\sim\) or
  %\(\mathord\prec = \mathord{\neg\sim}\).
\end{definition}
% Calling the \rebuild function on a language state in which some consistency guard has been inferred restructures the subtree rooted in this state to be complete, consistent and well guarded. It simply rebuilds the tree from its root, using the \findguard (Algorithm~\ref{findguard}) function to deduce the necessary guards. Following~\cite{GJP06}, we advise that the adjacent pairs and validity guards are considered in their order of introduction, to avoid large remodelling of the subtree.
\rebuild constructs a complete, consistent and well-guarded subtree if it is called high enough in the tree. 
\begin{restatable}{proposition}{rebuildsound}%{Soundness of \rebuild}{}
  Running \rebuild on a valid and consistent state~\(\stl\) of which no
  successors have inconsistencies that lead to consistency guards at a
  depth lesser than \(|\stl|\), constructs a subtree rooted in its
  argument that is \emph{complete}, \emph{consistent} and
  \emph{well-guarded}. It furthermore does not have \emph{invalid}
  states.
\end{restatable}

This proposition tells us we can keep the timed decision graph
up-to-date with respect to observations (\ie,~complete and consistent)
while keeping the good properties that were ensured by the previous
algorithms.  It~remains to show how a candidate timed automaton can be
constructed from this structure.

%% file: rebuild.tex
%We have so far discussed of observations structures. It is the core of our contribution, but remains an intermediary structure aiming to be reshaped into a RERA.

Following the active learning approach, our purpose is to identify a subset of nodes in the decision graph that will correspond
to locations of the automaton, and then fold transitions according to
an order on the remaining nodes.~\cite{GJP06} discusses such orders when resets are fixed. To handle RERA we first have to fix a \emph{reset strategy} before applying the original method. This gives as many hypotheses as we have strategies.
%% phrase incomplete. Je commente !
%s all constructed
%automata generalize the observation, the most obvious metric to compare the resulting automata is their size.

% We follow the lines of the active-learning approaches: their core is to
% identify a subset of nodes in the decision graph that will correspond
% to locations of the automaton, and then fold transitions according to
% an order. In~\cite{GJP06}, such order are introduced and discussed when
% in the case of ERA, for which no reset choices have to be made.
% %(\ie~for the ERA framework they consider).
% When generalizing to RERA, the~main difference is that we
% have a lot of different \emph{reset strategies} to pick~from. Once a 
% reset strategy has been fixed, the~original method can be applied, and the results
% can be compared to pick the most satisfying automaton. As all constructed
% automata generalize the observation, the most obvious metric to compare the resulting automata
% %select an hypothesis is its size,
% is their size, and it is the one we use here.
% We do not detail the exact orders that can be used,
% and refer the reader to \cite{GJP06} for more precisions. 
%\LH{J'ai complètement coupé la phrase.}
%\lNM{Pour l'instant, on ne voit pas trop \`a quoi servent ces ordres, donc cette phrase arrive trop tot, je pense. [apres lecture de la suite : clairement, il vaut mieux en parler a la section suivante]}

\paragraph{Reset selection.}
We present the general framework but do not discuss good strategies in
the following. Such strategies would rely on heuristics.
% We define here the reset strategies that we use to handle the
% selection of the reset through the tree.
\begin{definition}%{Reset strategy}{reset_strat}
  A \emph{reset strategy} over a timed decision graph \(\tdg\) is a
  mapping \(\rstrat\colon\Std\rightarrow \{\true,\false\}\), assigning
  a decision to each decision states.
\end{definition}
A reset strategy~$\pi$ is said \emph{admissible}
%\lLH{Correct? Maybe admissible conflicts a bit with the Game Theory usual def}
if for any state \(\std\), there is a language state \(\stl\) such
that \((\std,\rstrat(\std),\stl)\in\Trans\).
\begin{restatable}{proposition}{arsexists}%{Existence of admissible reset strategy}{adm_res_strat}
In a timed decision graph constructed using the 
\findpath and \rebuild
%\LH{J'ai mis les deux principaux. Il manquait surtout un 's' vu qu'on ne donne pas d'algo général (on pourrait mais la place?)}
algorithms
%\NM{Which one(s)?}
and
where every scheduled call to \rebuild has been done, there always
exists at least one admissible reset strategy.
\end{restatable}
An admissible reset strategy is used to prune the decision graph in such
a way that only one reset combination is considered for each transition.
% We~then define the graph resulting from the application of an admissible
% reset strategy to a~TDG.
%
%\begin{definition}%{Result of an admissible strategy}{res_graph}
The effect of an admissible reset strategy \(\rstrat\) on its timed
decision graph~\(\tdg\) is the TDG \(\rstrat(\tdg)\) defined
from~\(\tdg\) by keeping only outgoing transitions from decision
states that agree with~\(\rstrat\).
% by adding the following constraint on \(\Trans\):
% \[
%   \inference{(\std,r,\stl)\in\Trans}{r=\rstrat(\std)}.  
%\]
%\lNM{En fait, c'est une question d'habitude, je pense, mais je ne comprends pas trop pourquoi tu préfères écrire
%  des règles d'inférence plutôt que ``on ne garde dans $E$ que la transition $(\std,\pi(\std),\stl)$.'' La pr\'esentation est bien assez formelle comme \c ca... Aussi, je sugg\`ere de ne pas mettre cette construction dans une ``D\'efinition'' formellle...}
%\lLH{Comme ça? Question d'habitude de fait. Je trouve qu'en anglais c'est un peu trop "avec les mains" et que réintroduire les variables prend plus de place. Excès de cours de sémantique pendant mon master je le crains... mais je me soigne! :)}
We call this TDG the \emph{resulting
%\NM{terminologie classique ? ``resulting'' est un peu trop ``commun'', je trouve}\LH{Pas du tout, mais je n'avais pas mieux}
graph} of~\(\rstrat\). 
%\end{definition}
It can be seen quite directly that a resulting graph always has
exactly one successor to each decision state. Using this, we~can
notice that those resulting graphs are very close
%\NM{pr\'eciser la diff\'erence ?}
to timed decision trees of~\cite{GJP06}, in which no
decision states exist and the transitions from language states to
language states directly hold the (only possible) reset.

\paragraph{Orders and folding.}
%\lLH{This subsection is (almost) straight from Grinshtein's work.}
Once an admissible reset strategy is fixed, it is possible to fold the resulting graph into a RERA. This is made through the use of a preorder on states: we want to find a maximal subset for this order.

We define the \emph{height} of a language state \(\stl\), noted \(\height(\stl)\), as the height of the subtree it is the root of. A~preorder \(\preleq\) on language states is said \emph{height-monotone} when \(\stl\preleq\stl'\) implies \(\height(\stl)\leq\height(\stl')\). 
\begin{definition}
    Let \(\tdg\) be a timed decision graph and \(\preleq\) a preorder
    on its language states. A prefix-closed subset \(U\) of \(\tdg\)
    is called 
    %\begin{description}
        %\item[complete] if for all \(\stl\in U\) and \(a\in\Sigma\) there is a language state \(\stl.(a,g,r)\in\tdg\).
        \emph{\(\preleq\)-closed} if \(\stl\preleq U\) for all successors of \(U\) and 
        \emph{\(\preleq\)-unique} if for all \(\stl,\stl'\in U\), \(\stl\neq\stl'\Rightarrow\neg(\stl\preleq\stl')\).
     %\NM{ca veut dire que tous les \'el\'ements de $U$ sont incomparables ? Il vaut mieux le dire comme \c ca, non ? Est-ce que $\preleq$-unique est la terminologie ``officielle'' ?}\LH{J'ai repris ça à \cite{GJP06} sans trop faire gaffe. De fait "incomparable" c'est mieux!}
    %\end{description}
    %\lLH{J'ai enlev\'e la compl\'etude: pas besoin.}
\end{definition}
\(\preleq\)-closedness is used to construct a RERA by folding the successors of \(U\) into comparable states of \(U\). \(\preleq\)-uniqueness is useful to bound the number of states in~\(U\) and thus the size of the resulting automaton. 

The following lemma (Lemma~6.2 in~\cite{GJP06}) ensures that there always
exists a satisfying set of states~\(U\). For its constructive proof, we refer the
reader to the original paper.
\begin{lemma}
    Let \(\preleq\) be a height-monotone preorder on states in a resulting graph \(\rstrat(\tdg)\). Then there exists a \(\preleq\)-closed and \(\preleq\)-unique prefix-closed subset of the language states of \(\rstrat(\tdg)\). 
%\NM{Tu n'utilises pas ``complete'' ?}
\end{lemma}
Using such a subset, we can fold the resulting graph into a RERA as follows: 
\begin{definition}%{Merging}{}
    Let \((\Obs,\tdg)\) be a consistent observation structure, \(\rstrat\) an admissible reset strategy and \(\preleq\) a preorder on language states of \(\rstrat(\tdg)\). Consider a \(\preleq\)-unique, \(\preleq\)-closed and prefix-closed subset~$U$ of \(\rstrat(\tdg)\). Then a \emph{\(U_{\preleq}\)-merging} of \((\Obs,\tdg)\) according to \(\rstrat\) is a RERA 
    \(\tuple{U,\epsilon,\Clocks,\Trans,\Accept}\) such that %: 
    %\begin{itemize}
        %\item 
        \(\Accept=\{u\in U\mid \stlabel(u)=\{\ltrue\}\}\) and %,
        %\item 
        for any language node \(u.(a,g,r)\) of \(\rstrat(\tdg)\) with \(u\in U\), there is exactly one edge of the form \((u,a,g,r,u')\in\Trans\) with \(u.(a,g,r)\preleq u'\).
    %\end{itemize}
    Notice that, by the second condition, a \(U_{\preleq}\)-merging RERA is deterministic.
\end{definition}
Furthermore, if the observation structure is complete, a \(U_{\preleq}\)-merging generalizes the observations obtained so far. 

\paragraph{Constructing a candidate RERA.}
Using the results of the previous subsections, we can now construct a
candidate RERA from our observation structure. 
All~admissible reset strategies can be constructed by branch and bound. 
Then a merging is constructed for each resulting graph, and equivalence queries are launched.  
%\LH{Enlevé un paragraphe sur l'optimisation. On en est pas là\dots}

% For this, first notice
% that admissibility of a strategy is a \emph{local} property, so that all
% admissible strategies can be constructed by a branch-and-bound
% approach on the decision graph. It is then possible to construct a
% merging for each of the resulting graphs. All of those merging can
% then be analysed by an equivalence query.
% As equivalence queries typically have high cost, it~is interesting to
% use them with parsimony.
% %restrict the number we perform.
% For this, one can look at the size of the set \(U\)
% constructed for each resulting graph. This set impacts the size of the
% resulting automaton, thus the smallest \(U\)
% %of minimal size
% will be of interest.

For each of the RERA constructed by merging, either a counter-example
will be returned by the equivalence query, or the candidate is deemed
correct. In the latter case, we return this RERA; in the former case,
we include the counter-example in our observation structure and repeat
the process.

%% file: conclu.tex
In this paper, we propose an active learning method for deterministic reset-free event recording automata. 
We add a key feature to the state of the art:
% Our approach builds on previous works on deterministic event recording automata by adding a key feature: 
invalidity, that allows to detect incorrect guesses of resets when they are not tied to observations.
This required to rework all the data structures and algorithms involved to 
% This adds an important new complexity, and calls to rethink the data structure and\LH{juste les algos? La structure c'est plus simple} all the algorithms to 
handle invalidity on the fly. 
Most importantly, this brings the lacking notion to scale up to the class of deterministic timed automata (DTAs).
%, that requires only more computation power, and thus more profiled algorithms. 

A clear future work is to generalize this method to actually handle DTAs. This mostly requires to handles resets of sets of clocks instead of single ones. As the complexity would be greatly increased, this calls for some optimization. 
An promising addition would be to use an implicit structure. Instead of storing all possible reset configurations, only storing a small set of them at the same time would decrease the memory cost. As the models are built directly from observations, and not from previous states, the computational overhead may be limited.
An other interesting trail for future development is to find a way to build a timed automaton from the observation structure that exploits the different admissible reset strategies without building all of them. Works on approximate determinization of timed automata through games \cite{BJSK11} deal with similar problems and offer interesting leads. 
Finally, in \cite{GJP06}, the authors propose to refine the adjacent pairs into \emph{critical pairs}, that have a minimal set of differences. 
This allows to better identify the guards to be added, and thus can have a positive effect on both the size of the constructed models and the computational cost. Sadly, no precise procedure is given to construct the pairs, so creating one would be beneficial to the approach. 
More generally, studying the efficiency of this algorithm and of the variants proposed as future work could help better understand the applicability and bottlenecks of the approach.

%% file: appendixalgos.tex
%\subsection{Algorithms of section~\ref{sub:new_obs}}
\label{appendix-algos}
\begin{algorithm}
    \SetAlgoLined
    \caption{Adding a new observed timed word in \(\tdg\)}
    \label{findpath_tdg}
    {\sffamily \(\findpath_\tdg\)}\;
    \KwIn{a timed word \(w_t\) and its observation \(o\), its past \(p_t\) 
    , a valuation \(v\) and a language state \(\stl\)}
    \nl \eIf{\(w_t=\epsilon\)}{
    \nl   add \(o\) to the set of labels of \(\stl\)
        }{
    \nl   \((t,a).w'_t=w_t\)\;
    \nl   \eIf{\(\exists(\stl,g,a,\std)\in \Trans\)}{
    \nl   \For{\((\stl,g,a,\std)\in \Trans\)}{
    \nl     \If{\(v+t\models g\)}{
    \nl       \For{\((\std,r,\stl')\in \Trans\)}{
    \nl         \eIf{\(r=\true\)}{
    \nl           \(\findpath_\tdg(w'_t,o,p_t.(t,a),(v+t)_{[x_a\leftarrow0]},\stl')\)
                }{
    \nl           \(\findpath_\tdg(w'_t,o,p_t.(t,a),v+t,\stl')\)
                }
              }
    \nl       break
            }      
          }
          }{
    \nl   \(\addword_\tdg(w_t,o,p_t,\stl)\)
          }
      }
  \end{algorithm}
  \begin{algorithm}
    \SetAlgoLined 
    \caption{Adding a new observed timed word in \(\tog\)}
    \label{findpath_obs}
    {\sffamily \(\findpath_{\tog}\)}\\
    \KwIn{a timed word \(w_t\) and its observation \(o\), its past \(p_t\) and 
    reset history \(r\), a valuation \(v\) and an observation state \(\sto\)}
    \nl \eIf{\(w_t=\epsilon\)}{
    \nl   add \(o\) to \(\olabel(\sto)\)\;
    \nl   \If{\(|\olabel(\sto)|>1\)}{
    \nl   \(\sto.\invalid=True\)\;  
    \nl   \(\std=\parent(\sto)\)\;
    \nl   \While{all successors of \(\std\) are invalid}{
    \nl   \(\sto=\parent(\std)\)\;
    \nl   \(\sto.\invalid=True\)\;
    \nl   remove the last letters from \(p_t\) and \(r\)\;
    \nl   \(\std=\parent(\sto)\)\;
    }
    \(\searchprune(p_t,r,\stinit,\initv,\epsilon)\)}
    }{
    \nl   \((t,a).w'_t=w_t\)\;
    \nl   \eIf{\(\exists(\sto,g,a,\std)\in \Trans\)}{
    \nl   \For{\((\sto,g,a,\std)\in \Trans\)}{
    \nl     \If{\(v+t\models g\)}{
    \nl       \For{\((\std,r,\sto')\in \Trans\)}{
    \nl         \If{\(|\olabel(\sto')|=1\)}{
    \nl         \eIf{\(r=\true\)}{
    \nl           \(\findpath_{\tog}(w'_t,o,p_t.(t,a),r.\true,(v+t)_{[x_a\leftarrow0]},\sto')\)
                }{
    \nl           \(\findpath_{\tog}(w'_t,o,p_t.(t,a),r.\false,v+t,\sto')\)
                }}
              }
    \nl       break
            }      
          }
          }{
    \nl   \(\addword_{\tog}(w_t,o,p_t,\sto)\)
          }
      }
  \end{algorithm}
  \begin{algorithm}
    \SetAlgoLined
    \caption{Extending \(\tdg\) to satisfy a new timed word}
    \label{addword_tdg}
    {\sffamily \(\addword_\tdg\)}\\
    \KwIn{a non-empty timed word \(w_t\), its observation \(o\), its past \(p_t\) 
    and a language state \(\stl\)}
    \nl   \((t,a).w'_t=w_t\)\;
    \nl   create \(\std=(\stl,a,\gtrue)\) in \(\Std\)\;
    \nl   \(w_s,z_s=\stl\)\;
    \nl   create \(\stl'=(w_s.(a,\gtrue),z_s\future)\)\; 
    \nl   create \(\stl''=(w_s.(a,\gtrue),{{z_s}_{[x_a\leftarrow0]}}\future)\)\;
    \nl   create \((\stl,a,\gtrue,\std)\), \((\std,\false,\stl')\) and 
    \((\std,\true,\stl'')\) in \(\Trans\)\;
    \nl   \eIf{\(w'_t=\epsilon\)}{
    \nl     label \(\stl'\) and \(\stl''\) by \(\{o\}\)
    }{
    \nl     \(o'=Request(p_t)\)\;
    \nl     label \(\stl'\) and \(\stl''\) by \(\{o'\}\)\;
    \nl     \(\addword_\tdg(w'_t,o,\stl')\)\;
    \nl     \(\addword_\tdg(w'_t,o,\stl'')\)
    }
  \end{algorithm}
  \begin{algorithm}
    \SetAlgoLined
    \caption{Extending \(\tog\) to satisfy a new timed word}
    \label{addword_tog}
    {\sffamily \(\addword_{\tog}\)}\\
    \KwIn{a non-empty timed word \(w_t\), its observation \(o\), its past \(p_t\) 
    and an observation state \(\sto\)}
    \nl   \((t,a).w'_t=w_t\)\;
    \nl   create \(s_d=(s_l,a,\reg(v+t))\) in \(\Std\)\;
    \nl   \(w_s,z_s=\sto\)\;
    \nl   create \(\sto'=(w_s.(a,\reg(v+t)),\reg(v+t))\)\; 
    \nl   create \(\sto''=(w_s.(a,\reg(v+t)),\reg(v+t)_{[x_a\leftarrow0]})\)\;
    \nl   create \((\sto,a,\reg(v+t),\std)\), \((\std,\false,\sto')\) and 
    \((\std,\true,\sto'')\) in \(\Trans\)\;
    \nl   \eIf{\(w'_t=\epsilon\)}{
    \nl     label \(\sto'\) and \(\sto''\) by \(\{o\}\)
    }{
    \nl     \(o'=Request(p_t)\)\;
    \nl     label \(\sto'\) and \(\sto''\) by \(\{o'\}\)\;
    \nl     AddWord(\(w'_t,o,\sto'\))\;
    \nl     AddWord(\(w'_t,o,\sto''\))
    }
  \end{algorithm}
  \begin{algorithm}
    \SetAlgoLined
    \caption{Requesting an observation.}
    \label{request}
    {\sffamily \request}\\
    \KwIn{A timed word \(w_t\).}
    \KwOut{A unit label in \(\{\{\ltrue\},\{\lfalse\}\}\)}
  \nl  \eIf{\(w_t\in \Dom(\Obs)\)}{
  \nl    \textbf{return} \(\Obs(w_t)\)
    }{
  \nl    make an membership query on \(w_t\) and add its result \(o\) to \(\Obs\)\;
  \nl    \(\findpath_{\tog}(w_t,o,\epsilon,\initv,\stinit)\)\;
  \nl    \textbf{return} \(o\)
    }
  \end{algorithm}
  \begin{algorithm}
    \SetAlgoLined
    \caption{Pruning \tdg~after detecting an invalid timed word with resets.}
    \label{searchprune}
    {\sffamily \searchprune}\\
    \KwIn{A non-empty timed word \(w_{t}\) and a set of resets \(r\) of same length, 
    a language state \(\stl\), a valuation \(v\) and an history 
    \(h_d\in(\Std\times\{\true,\false\})^{*}\).}
    \nl \((t,a).w_t'=w_t\)\;
    \nl \(r_a.r'=r\)\;
    \nl \For{\((\stl,g,a,\std)\in\Trans\)}{
    \nl \If{\(v+t\models g\)}{
    \nl \For{\((\std,r_a',\stl')\in\Trans\)}{
    \nl \If{\(r_a'=r_a\)}{
    \nl \eIf{\(w'_t=\epsilon\)}{
    \nl remove \((\std,r_a',\stl')\) from \(\Trans\) and recursively delet the subtree\;
    \nl \If{\(|\{(\std,\_,\_)\in\Trans\}|=0\)}{
    \nl Schedule the subtree rooted in \(\stl\) to be rebuilt.
        }
        }{
    \nl          \eIf{\(r_a=\true\)}{
    \nl          \(\searchprune(\stl',(v+t)_{[x_a\leftarrow0]},h_d.(\std,r_a))\)
                 }{
    \nl            \(\searchprune(\stl',(v+t),h_d.(\std,r_a))\)
              }
            }
    \nl        break
          }
        }
    \nl    break
      }
    }
  \end{algorithm}

%\subsection{Algorithms of Sec.~\ref{sub:incons}}

  \begin{algorithm}
    \SetAlgoLined
    \caption{Finding an adjacent pair corresponding to an inconsistency.}
    \label{adjpair}
    {\sffamily \adjpair}\\
    \KwIn{A timed word with resets \(w\) such that \(\Obs(w_{tr})=\ltrue\) and a second one 
    \(w'\) such that \(\Obs(w')=\lfalse\).}
    \KwOut{An adjacent pair}
    \nl \While{\(\pnot(\forall i\forall a,\ |v_i(x_a)+t_i-(v'_i(x_a)+t'_i)|<1 \vee
    ((v_i(x_a)+t_i>K)\wedge(v'_i(x_a)+t'_i>K)))\)}{
    \nl \(w''=0.5w+0.5w'\)\\
    \nl \eIf{\request(w'')}{
    \nl \(w=w''\)
    }{
    \nl \(w'=w''\)
    }
    }
    \nl\For{\(i\in[0,n]\), \(a\in\Sigma\)}{
    \nl  \If{\(\floor{v_i(x_a)+t_i}\neq\floor{v'_i(x_a)+t'_i}\wedge v_i(x_a)+t_i\not\in\bbN
      \wedge v'_i(x_a)+t'_i\not\in\bbN\)}{
    \nl    \eIf{\(v_i(x_a)+t_i < v'_i(x_a)+t'_i\)}{
    \nl      \(\lambda=(\floor{v'_i(x_a)+t'_i}-(v_i(x_a)+t_i))/(v'_i(x_a)+t'_i-(v_i(x_a)+t_i))\)\\
    \nl      \(w''=\lambda.w'+(1-\lambda).w\)
        }{  
    \nl      \(\lambda=(\floor{v_i(x_a)+t_i}-(v'_i(x_a)+t'_i))/(v_i(x_a)+t_i-(v'_i(x_a)+t'_i))\)\\
    \nl      \(w''=\lambda.w+(1-\lambda).w'\)
        }
    \nl    \eIf{\(\request(w'')=\ltrue\)}{
    \nl      \(w=w''\)
        }{
    \nl      \(w'=w''\)
        }
      }
    }
    \nl \(w''=0.5w+0.5w'\)\\
    \nl \eIf{\(\request(w'')=\ltrue\)}{
    \nl  return \((w',w'')\)
    }{
    \nl  return \((w,w'')\)
    }
  \end{algorithm}
  
%\subsection{Algorithms of Sec.~\ref{sub:inval}}
We define common borders for the use of the \invalguard algorithm.
\begin{definition}
    For two K-equivalence classes \(z,\ z'\), we define their \emph{common borders} 
  \(\cb(z,z')\subset\Clocks\times\{<,\leq,\top\}\times\bbN\) as: 
  \[
  \inference[]{z\subset x=k,\ z'\subset k-1<x<k}{(x,<,k)\in\cb(z,z')}
  \inference[]{z\subset x=k,\ z'\subset k<x<k+1}{(x,\leq,k)\in\cb(z,z')}
  \]\[
  \inference[]{z\subset k-1<x<k+1,\ z'\subset k<x<k+1}{(x,\top,k)\in\cb(z,z')}
  \]
  and the same rules, inverting $z$ and $z'$.
  \end{definition}
  The common borders are used both to measure the proximity of K-equivalence classes and to create validity guards.
\begin{algorithm}
    \SetAlgoLined
    \caption{Finding a guard to separate two invalidities.}
    \label{invalguard}
    {\sffamily \invalguard}\\
    \KwIn{two sets $W_1$ and $W_2$ of timed words with resets corresponding to invalidities. Each set of words passes through a \(w_{gr}\) and then shares a same K-equivalence class to play a common action $a$. $W_1$ is invalid for reset $\true$ while $W_2$ is invalid for reset $\false$.}
    \KwOut{A validity guard}
    \nl \(n=|w_{gr}|\)\\
    \nl We note the shared K-equivalence classes to play action $a$ \(z^1_n\) for \(W_1\) and \(z^2_n\) for \(W_2\)\\
    \nl \(W'_1,W'_2=\emptyset\)\\
    \nl \eIf{\(\cb(z^1_n,z^2_n)=\emptyset\)}{
    \nl \For{\(w_t^1\in W_1\))}{
    \nl \For{\(w_t^2\in W_2\)}{
    \nl \(w=0.5w_t^1[1,n+1]+0.5w_t^2[1,n+1]\)\\
    \nl \({w'}_t^1=w\op w_t^1\) and \({w'}_t^2=w\op w_t^2\)\\
    \nl \request(${w'}_t^1$); \request(${w'}_t^2$)\\
    \nl \(W'_1+=\{{w'}_t^1\}\), \(W'_2+=\{{w'}_t^2\}\)
    }
    }
    \nl \If{the new observations have made \(w_{gr}\) invalid}{
    \nl   stop
    }
    \nl \If{\(w\) completed with the resets of \(w_{gr}\) plus \(\true\) is invalid}{
    \nl   \invalguard(\(W'_1,W_2,w_{gr},a\))
    }
    \nl \If{\(w\) completed with the resets of \(w_{gr}\) plus \(\false\) is invalid}{
    \nl   \invalguard(\(W_1,W'_2,w_{gr},a\))
    }
    \nl \If{\(w\) completed with the resets of \(w_{gr}\) plus \(\true\) and \(\false\) are valid}{
    \nl   \invalguard(\(W'_1,W_2,w_{gr},a\)) \invalguard(\(W_1,W'_2,w_{gr},a\))
    }}{
      \nl we call \(\stl\) the language state in which \(w_{gr}\) ends.\\
      \nl return \((\stl,a,x,\sim,k)\) where \((x,\sim,k)\in\cb(z^1_n,z^2_n)\)
    }
  \end{algorithm}
  
  %\subsection{Algorithms of Sec.~\ref{sub:rebuild}}
  \begin{algorithm}
    \SetAlgoLined
    \caption{Rebuilds a subtree of \tdg~to handle consistency.}
    \label{rebuild}
    {\sffamily \rebuild}\\
    \KwIn{A (valid) language state \(\stl\)}
    \nl Suppress recursively all successors of \(\stl\)\;
    \nl \For{\(a\in\Sigma\) such that there is at least an observation passing \(\stl.(a,\true)\)}{
    \nl \For{each guard \(g\in\findguard(\stl,a,\true)\)}{
    \nl create \(\std=(\stl,a,g)\) and \((\stl,a,g,\std)\in\Trans_\Sigma\)\;
    \nl \If{\(\stl.(a,g,\true)\) is not invalid}{
    \nl create \(\stl'=\stl.(a,g,\true)\)\;
    \nl \(\stlabel(\stl'')=\request(\stl'')\)\;
    \nl \rebuild(\(\stl'\))
    }\nl \If{\(\stl.(a,g,\false)\) is not invalid}{
    \nl create \(\stl''=\stl.(a,g,\false)\)\;
    \nl \(\stlabel(\stl'')=\request(\stl'')\)\;
    \nl \rebuild(\(\stl''\))
    }
    }
    }
  \end{algorithm}
  \begin{remark}%{Use of \request in \rebuild.}{}
    In the rebuild function, we use \request on guarded words with resets instead of timed word. The extension is quite simple thanks to the resets, as searching in \tog~if an observation modelling the argument exists is only a dive in the tree, and if none is found, making an membership query from the last guess is the same.
  \end{remark}
  \begin{remark}
    As written, \rebuild completely erases the subtree and then reconstructs it. An obvious optimization is to only suppress transitions and nodes when necessary to avoid invalidity or inconsistency. We do not develop this here to keep the algorithm short and simple.
  \end{remark}
  \begin{algorithm}
    \SetAlgoLined
    \caption{Find a partition in guards to be applied to an action in a language node}
    \label{findguard}
    {\sffamily \findguard}\\
    \KwIn{A (valid) language state \(\stl\), an action \(a\) and a guard \(g\)}
    \KwOut{a partition of \(g\)}
    \nl \eIf{there is an adjacent pair passing \(\stl.(a,g)\) from which a consistency guard at depth \(|\stl|\) can be deduced}{
    \nl let g' be this guard\;
    \nl  \textbf{return} \(\findguard(\stl,a,g\wedge g')\cup\findguard(\stl,a,g\wedge \neg g')\)
    }{
    \nl \eIf{\(\stl,g,\true\) or \(\stl,g,\false\) is not invalid}{ 
    \nl \textbf{return} \(\{g\}\)
      }{
    \nl let \(\stl,a,g',x,\sim,k\) for \(g\subseteq g'\) be a validity guard that differentiate the invalidities.\\
    \nl Let \(g''=x<k\) if $\sim=<$ and \(g''=x\leq k\) otherwise\\
    \nl  \textbf{return} \(\findguard(\stl,a,g\wedge g'')\cup\findguard(\stl,a,g\wedge \neg g'')\)
      }
    }
  \end{algorithm}

%% file: appendixproofs.tex
We conduct here the proofs of different claims. 

\subsection{Proofs of section \ref{sec:prelim}}
\clockcomb*
\begin{proof}
    The proof is made by induction on \(w_t^3\). For \(\epsilon\), the only valuation 
    encountered is \(v_{0,r}^3=\initv=\lambda.v_{0,r}^1+(1-\lambda).v_{0,r}^2\). 
    For \(w_t^j={w'^j_t}.(t^j,a)\), \(j\in[1,3]\), assume that we have the property for 
    all valuations reached along the prefixes, and especially that for all clocks 
    \(v'^3_r(x_a)=\lambda.v'^1_r(x_a)+(1-\lambda).v'^2_r(x_a)\) for the last valuations 
    encountered along the \(w'^j_t\). Then we have that 
    \(v'^3_r+t_3={v'^3_r}+\lambda.t^1+(1-\lambda).t^2=\lambda.(v'^1_r+t^1)+(1-\lambda).(v'^2_r+t^2)\) 
    and by applying the resets commended by $r$ we obtain the result desired for the last 
    valuation.  
\end{proof}

\subsection{Proofs of section \ref{sec:structure}}
We prove formally that TDGs and TOGs are trees.
\begin{proposition}%{}{tdg_tree}
  The part of a timed decision graph reachable from~$s_0$ is a bipartite tree.
\end{proposition}
\begin{proof}
    Consider a timed decision graph \(\tdg\).
    By definition of \(\Trans\), \(\tdg\) is bipartite. It is acyclic because any path 
    from a language state to an other one leads to a state a guarded word with resets of 
    strictly greater length. And any given state has exactly a unique predecessor, but 
    \(\stinit\) that has none. Indeed, for a decision state \(\std=(\stl,a,g)\) the only possible 
    predecessor is the language state \(\stl\), by definition of \(\Trans\); for 
    a language state different from \(\stinit\), \(\stl=(w.(g,a,r),z)\) the only 
    possible king of incoming transition in \(\Trans\) is \(\std=((w,z'),a,g),r,\stl\) with an a priori unknown \(z'\). But as 
    we dispose in \(w\) of the sequence of guards, and precise resets that occurred, 
    we can inductively compute \(z'\) by iterating taking the future of the current zone, 
    intersecting it with the guard and applying the desired reset, starting from 
    the zone in \(\stinit\). Thus there is only one possible \(z'\), and a unique predecessor to \(\stl\).
\end{proof}

\prsound*
\begin{proof}
By coverage, there are two timed words \(w_t, w'_t\in \Dom(\Obs)\) such that 
\(w_t,w'_t\models w_{zr}\) and \(\Obs(w_t)\neq \Obs(w'_t)\). For any given timed automaton, 
if \(w_{zr}\) is a subset of executions, then as every zone in this word is a K-equivalence class, 
it corresponds to the same sequence of pairs of locations and K-equivalence classes. It then comes that 
\(w_t\) and \(w'_t\) are both executions leading to equivalent configurations. Hence 
the automaton fails to model the differences of the corresponding observation, as such 
configurations have the same location, which can't be both accepting and not accepting.
\end{proof}

\subsection{Proofs of section \ref{sec:updates}}
\fptdgsound*
\begin{proof}
First of all, a call to \(\findpath_\tdg\) terminates, as recursive calls are in finite number and on 
words of strictly decreasing length, there is only 1 call to \(\addword_\tdg\) from \(\findpath_\tdg\) along 
each explored path and recursive calls to \(\addword_\tdg\) are in finite number and on words of 
strictly lower length. 

We now prove the rest of the property by induction on the calls to both \(\findpath_\tdg\) and \(\addword_\tdg\). 
We use the following induction hypothesis: 
A call to \(\addword_\tdg\)/\(\findpath_\tdg\) creates a subgraph that is complete (and takes into account 
the new observation) and such that for all reachable language states \(\stl\) in that 
subgraph, \(\mid label(s_l) \mid\geq 1\).
\begin{itemize}
    \item Basic case for \(\findpath_\tdg\). Here we have \(w_t=\epsilon\). In this case the 
    complete word to add was read before along the path, and the only performed action 
    is to add the observation to the label of \(\stl\). Hence the subgraph is complete with 
    respect to the new observation (by adding delays and actions no other reachable states 
    can correspond to that same word). No other states are created, hence the hypothesis 
    on the initial observation structure suffices to conclude that the subgraph is 
    complete and verifies that all labels have at least one element. 
    \item Basic case for \(\addword_\tdg\). As \(\addword_\tdg\) is never called on empty 
    words, we have \(w_t=(t,a)\). The call to AddWord adds a new decision state 
    \(\std=(\stl,a,\true)\) and two new language nodes \(stl'\) and \(\stl''\) corresponding
    to the effect of resetting or not \(x_a\) after the action. As \(\addword_\tdg\) 
    was called in \(\findpath_\tdg\), we know that no successors for action \(a\) existed 
    in \(\stl\) (as the graph is well defined and thus if one existed, one would have 
    covered \(w_t\)). As we are in the base case, the labels of \(\stl'\) and 
    \(\stl''\) are augmented with the observation \(o\), making them have a non-empty label. 
    The edges constructed by this call are in accord with the definition, and by the 
    hypothesis on the initial observation structure, the other successors of \(\stl\) are 
    complete (except with respect to the new word, but their sequence of letters 
    do not match) and have non-empty labels, hence in all cases we obtain a subgraph that is 
    complete (as \(\stl'\) and \(\stl''\) have no successors) and have only non-empty labels. 
    \item Inductive case of \(\findpath_\tdg\). We consider that \(w_t=(t,a)w'_t\). 
    Thus we enter the else in line 3. If the else case is called in line 12, we only make a 
    call to \(\addword_\tdg\) on \(w_t\) hence by induction  hypothesis, we have the 
    desired properties. Else, as there is only one guard such that \(w_t\) can go through 
    that guard, all successors satisfying a prefix of \(w_t\) are reached by the recursive 
    calls and by induction hypothesis they lead to complete subgraphs with non-empty labels. 
    Furthermore, other successors of \(\stl\) constitute, by the hypothesis on the initial 
    structure, a complete subgrap (except that the guard \(g\) of the considered transition 
    is not covered) with only non-empty labels, except for the new word that may not 
    be covered. But by uniqueness, they can not correspond to paths satisfying the new word 
    and thus we have our properties. 
    \item Inductive case for \(\addword_\tdg\). This case works exactly as the base case, 
    except that calls to "request" ensure that the new states have non-empty labels, and 
    the completeness with respect to \(w_t\) is ensured by the induction hypothesis. 
\end{itemize}
\end{proof}

\fptogsound*
\begin{proof} 
As for the proof of Prop.\ref{pr:findpath_tdg_sound} termination is clearly ensured by 
the structure of recursive calls. Notice that we do not count in this the calls to 
\(\findpath_{\tog}\) made in \(\request\), as they deal with different words. 
The same kind of induction on calls of \(\findpath_{\tog}\) and 
\(\addword_{\tog}\) suffices to prove correspondence and coverage if no ancestor of the state has a label of cardinality two. If one has, then by definition the state is invalid.
\end{proof}

\spsound*
\begin{proof}
Invalidity is detected along the calls to \(\findpath_{\tog}\), and the propagation of the invalid tag follows the definition for all ascendant states. No descendant are tagged, but this does not matter as they can not be reached without passing by invalid states as \(\tog\) is a tree. A call to \(\searchprune\) is then made, that targets exactly the root of the invalid subtree that has been detected, and prune it. As this is made for all detected invalidities and the subtrees are detected, when the procedure terminates, no language state invalid because of an invalidity detected in \(\findpath_{\tog}\) can be reached. Furthermore only the invalid subtree is suppressed, hence no valid state is made unreachable (as by definition all descendant of invalid states are invalid).

To conclude, it suffices to notice that every new membership query gives rise to a corresponding call to \(\findpath_{\tog}\), leaving no invalidity undetected.
\end{proof}

\rebuildsound*
\begin{proof}
    We prove these four properties independently.
    \begin{description}
      \item[Well-guardedness] The well-guardedness comes directly from \findguard, as only consistency guards corresponding to passing adjacent pairs and validity guards are added to the guards.
      \item[Validity] Validity comes from the validity of \(\stl\) (by hypothesis) and the validity test made for all descending language trees. Notice that there is also always a successor to any decision state, as the parent language state is valid, and \findpath ensures to construct a guard leaving a reset configuration open.
      \item[Consistency] The label of each created state receive an element, so it can not be empty. Furthermore, as each inconsistency in the original subtree rooted in \(\stl\) lead to a consistency guard of depth greater than \(|\stl|\) and \findguard adds all those consistency guards to the required path, no state can have a label of cardinality two. Thus, combined with the hypothesis on \(\stl\), the subtree is consistent.
      \item[Completeness] As the \rebuild function continuously calls itself as long as an observation passes the current word, all observations model a word. The condition of label is ensured by the consistency proof: as each state has a non-empty label and each inconsistency has been split by \findguard, the label of each observation is in the label of the state it models.
    \end{description}
\end{proof}

\subsection{Proofs of section \ref{sec:rebuild}}

\arsexists*
\begin{proof}
    To ensure that an admissible reset strategy exists, one only needs to check that every decision state has at least one successor. We only prune the graph in the \searchprune algorithm, and this algorithm schedules a call to \rebuild when no successors exist for a decision state. As \rebuild constructs a subtree where all decision states have at least a successor (thanks to the \findguard function that explicitly checks for this), we have our property.
\end{proof}